\documentclass{article}
\usepackage{appendix}
\usepackage{amsmath}
\usepackage{multirow}

\usepackage{graphicx}
\usepackage{epstopdf}
\usepackage{extarrows}
\usepackage{float}
\usepackage[latin1]{inputenc}
\usepackage{tikz}
\usepackage{mathrsfs}
\usepackage{amssymb}
\usepackage{stmaryrd}
\usepackage{mathabx}
\usepackage{algpseudocode}
\usepackage[normalem]{ulem}
\usepackage{flushend}

\usetikzlibrary{trees}
\usepackage{tikz-qtree}
\usepackage{multirow}
\hfuzz=\maxdimen
\usepackage{tikz}
\usepackage{subfigure}
\usepackage{amsmath}
\usepackage{amssymb}
\usepackage{array}
\newcolumntype{C}[1]{>{\centering\arraybackslash}p{#1}}
\usetikzlibrary{arrows,automata}

\newtheorem{definition}{Definition}
\newtheorem{theorem}{Theorem}
\newtheorem{lemma}{Lemma}

\newtheorem{remark}{Remark}
\newtheorem{algorithm}{Algorithm}
\newtheorem{example}{Example}
\newtheorem{proposition}{Proposition}
\newenvironment{proof}{{\noindent\it \textbf{Proof}}\quad}{
\hfill\rule{2mm}{2mm}\par}

\begin{document}

\title{Supervisory Control of Probabilistic Discrete Event Systems under Partial Observation}

\author{Weilin Deng, Jingkai Yang,
        and Daowen Qiu$^{\star}$ 
\thanks{Weilin Deng is  with the School of Data and Computer Science, Sun Yat-sen University,
Guangzhou, 510006, China. (e-mail:williamten@163.com)}
\thanks{Jingkai Yang is  with the School of Data and Computer Science, Sun Yat-sen University,
Guangzhou, 510006, China. (e-mail:yangjingkai711@163.com)}
\thanks{Daowen Qiu (Corresponding author) is with the School of Data and Computer Science, Sun Yat-sen
 University, Guangzhou, 510006, China. (e-mail: issqdw@mail.sysu.edu.cn)}}

\maketitle

\begin{abstract}
 The supervisory control of probabilistic discrete event systems (PDESs) is investigated
 under the assumptions that the supervisory controller (supervisor) is probabilistic and has a partial observation.
   The probabilistic P-supervisor is defined, which specifies a probability distribution on the control patterns for each observation.
  The notions of the probabilistic controllability and observability are proposed and demonstrated to
  be a necessary and sufficient conditions for the existence of the probabilistic P-supervisors.
  Moreover, the polynomial verification algorithms for the probabilistic controllability and observability  are put forward.
  In addition,
  the infimal probabilistic controllable and observable superlanguage is introduced and computed as the solution of the optimal control problem of PDESs.
  Several examples are presented to illustrate the results obtained.
\end{abstract}

%
%

\section{Introduction}

Discrete event systems (DESs) are the event-driven systems with discrete states \cite{desbook}.
The supervisory control of DESs was initially proposed by Ramadge and Wonham \cite{first-supevisory}.
Since then, this problem  has gained extensive attention in community (see  e.g., [3]-[13]).
Supervisory control of the partial observed DESs was first considered by Lin and Wonham \cite{observability}, and Cieslak \emph{et al.} \cite{observability2}, respectively.
The necessary and sufficient conditions for an achievable specification by supervisory control
were presented in \cite{observability}.
When these conditions are not met, that is, the given specification is unachievable,
the issue of finding an achievable approximation come forward.
The infimal controllable and observable superlanguage of the specification was considered in \cite{Rudie1}, \cite{Shayman} and \cite{Masopust}.
Moreover,
the calculating algorithm for the supremal controllable normal sublanguage was provided in \cite{Cho},
and a larger controllable and observable sublanguage was obtained in \cite{Takai}.
In particular,
the synthesis issue of the maximally permissive supervisors for the partial observed DESs
 was also well investigated in \cite{Yin1} and \cite{Yin2}.

\par
However, the (conventional) DESs model cannot characterize the probabilistic properties of the probabilistic systems,
which exist commonly in the engineering field.
In order to characterize these probabilistic systems accurately,
the \emph{probabilistic discrete event systems} (PDESs) model (or called as the \emph{stochastic  discrete event systems} model in some literature), as an extension of the DESs model, was put forward.
Besides the components in the DESs model,
the PDESs model additionally defines another component concerning probabilities,
including the transition probabilities between states and the termination probabilities at states.
The probability component offers the PDESs model the ability to characterize the uncertainties  in  probabilistic systems.
\par

In recent years, the PDESs model  has received considerable attention in the community,
especially on the issues of supervisory control  (see e.g., [13]-[22]),
 fault diagnosis  (see e.g., [23]-[25]),
 fault prediction  (see e.g., [26]-[29]),
and state detection \cite{detectability1}, \cite{detectability2}.
Moreover, the PDESs model also has been applied to many practical problems in various areas \cite{SDES-book},
such as robot control [32]-[34],
tobacco control \cite{tobacco-control}, multi-risk systems \cite{multi-risk} and so on.

\par
In this paper, we are devoted into researching the  control of PDESs, which was considered by
Lin \emph{et al.} \cite{Lin-1}, \cite{Lin-2}, and Kumar \emph{et al.} \cite{Kumar}, and Lawford \emph{et al.} \cite{Lawford}, \cite{Pantelic-1}, respectively.
In supervisory control of PDESs, the specifications can be deterministic and probabilistic.
However, the control approaches for deterministic specifications, such as safety and non-blockingness, have been well investigated for (conventional) DESs, and
those approaches also can be applied to PDESs by only considering their logic parts.
Hence, only probabilistic specifications are dealt with in [14]-[18].
For simplicity, this paper also just handles  probabilistic specifications.

\par
Lin \emph{et al.} \cite{Lin-1}, \cite{Lin-2} investigated the issue of the fault-tolerant control of PDESs with ``soft" specifications,
which can be violated within a tolerable scope.
Lin \cite{Lin-1} first considered using probability to specify what is tolerable in the fault-tolerant control.
Li \emph{et al.} \cite{Lin-2}  further developed the work presented in \cite{Lin-1},
in which the authors  discussed standard supervisor  synthesis problem and reset supervisor  synthesis problem.

\par
Kumar \emph{et al.} \cite{Kumar}
investigated the ``range" control problem of PDESs with the assumption that the supervisor has a full observation.
The specification of the ``range" control is given by  a pre-specified range,
where the upper bound is a non-probabilistic language, and the lower bound is a probabilistic language.
Similar to \cite{Lin-1} and \cite{Lin-2},
\cite{Kumar} also defined a \emph{deterministic supervisor} that always issues determinate control actions.

\par
Lawford \emph{et al.}  \cite{Lawford} demonstrated that the probabilistic supervisor generates a much larger
class of probabilistic languages than the deterministic supervisor does.
Hence, they pointed out that the probabilistic control mechanism is much more powerful than the deterministic one.
As a result, Lawford \emph{et al.} \cite{Lawford}, \cite{Pantelic-1}, defined  a probabilistic supervisor,
which is also supposed to have a full observation.

 \par
 Recently,  Pantelic \emph{et al.} \cite{Pantelic-2-optimal-control} investigated the optimal control problem of PDESs.
 The optimal control aims to synthesize a supervisor that
 minimizes the distance between the uncontrollable specification and its controllable approximation.
  In order to measure the distance,
  Pantelic \emph{et al.} \cite{Pantelic-4} proposed the notion of the \emph{pseudometric}, and its calculating algorithms.

  \par
 Chattopadhyay \emph{et. al} \cite{optimal-control} also considered the optimal control issue of PDESs.
  However, different from \cite{Pantelic-2-optimal-control},
  the optimal objective is maximizing the renormalized language measure vector for the controlled plant.
  Based on the measurement, Chattopadhyay \emph{et. al} also formulated a theory for the optimal control of PDESs.

\par
It should be pointed out that the supervisors defined in [16]-[22] are all supposed to have a full observation to the events,
which are not always satisfied in practical engineering systems.
 In this paper, we focus on the supervisory control problem of PDESs
 with the assumptions that the supervisor is probabilistic and has a partial observation to the events.

 \par
 Different from the full-observation supervisors defined in [16]-[22],
 we define a partial-observation probabilistic supervisor, called as \emph{the probabilistic P-supervisor},
 which specifies a probabilistic distribution on the control patterns to each observation.
 Intuitively, for each observation, the probabilistic P-supervisor makes a special roulette.
 The roulette issues several outcomes with the pre-specified probabilities.
  Before making a control decision,  the supervisor will ``roll" the corresponding roulette, and then adopt the $j$th control pattern if
the $j$th outcome is issued.
In addition, we demonstrate the equivalence between the probabilistic P-supervisor and the scaling-factor function.

\par
We then present the notions of the probabilistic controllability and observability, and their polynomial verification algorithms.
We also demonstrate that the probabilistic controllability and observability are the necessary and sufficient conditions for the existence of the probabilistic P-supervisors, and provide the design method of the probabilistic P-supervisor.
Moreover, we consider the optimal control problem of PDESs.
Different from [20-22],
 the \emph{infimal probabilistic controllable and observable superlanguage} is defined and computed as the solution of the optimal control problem of PDESs.

    \par
   The rest of the paper is organized as follows.
 The related notations and the necessary preliminaries are presented in Section \uppercase\expandafter{\romannumeral2}. Then
   the probabilistic P-supervisor is defined, and the equivalence of the probabilistic P-supervisor and the scaling-factor function is demonstrated in Section \uppercase\expandafter{\romannumeral3}.
  After that, in Section \uppercase\expandafter{\romannumeral4}  the supervisory control theory of PDESs under the partial observations are established.
     Finally, we investigate the optimal control problem of PDESs in Section \uppercase\expandafter{\romannumeral5},
    and summarize the main results  and mention several future research directions in Section \uppercase\expandafter{\romannumeral6}. \par

\section{Notation and Preliminaries}
    In this section, we would introduce the automata model and languages model for probabilistic discrete event systems (PDESs).
  \par
    A PDES is usually characterized by a probabilistic automaton. Formally, we present the following definition.
    \begin{definition}
        A PDES could be modeled as the following probabilistic automaton:
        \begin{equation}
             G = \{ X, x_{0}, \Sigma, \delta, \rho  \}.
        \end{equation}
    \begin{itemize}
        \item
        $X$ is the nonempty finite set of states.

        \item
        $x_{0} \in X $ is the initial state.

        \item
        $\Sigma$ is the nonempty finite set of events. $\Sigma = \Sigma_{c} \cup \Sigma_{uc}$, where the $\Sigma_{c}$ and $\Sigma_{uc}$ denote the controllable and uncontrollable events set, respectively.
        Without loss of generality, in this paper, let $|\Sigma| = n$, $|\Sigma_{c}| = m$, and $\sigma_{i} \in \Sigma_{c}$ for  $i \in [1,m]$,
        and $\sigma_{i} \in \Sigma_{uc}$ for $i \in [m+1,n]$.
        Moreover, $\Sigma = \Sigma_{o} \cup \Sigma_{uo}$, where the $\Sigma_{o}$ and $\Sigma_{uo}$ denote the observable and unobservable events sets, respectively.

        \item
        $\delta: X \times \Sigma \rightarrow X$ is the (partial) transition function.
        The function $\delta$ can be extended to $X \times \Sigma^{*} $ by the natural manner.

         \item
        $\rho: X \times \Sigma \rightarrow [0,1]$ is the  transition-probability function.
        $\rho(x,\sigma)$ is the probability of the transition $\delta(x,\sigma)$.
        If $\delta(x,\sigma)$ is defined, denoted as $\delta(x,\sigma)!$, then $\rho(x,\sigma) > 0$;
        else if $\delta(x,\sigma)$ is not defined, denoted as $\delta(x,\sigma)$\sout{!}, then $\rho(x,\sigma) = 0$.
        Particularly, $\forall x \in X$, $\sum_{\sigma \in \Sigma} \rho(x,\sigma) \leq 1 $.
        $\sum_{\sigma \in \Sigma} \rho(x,\sigma)$, called as \emph{the liveness of the state} $x$,
        characterizes the possibility of certain events occurring at the state $x$.
         $1 - \sum_{\sigma \in \Sigma} \rho(x,\sigma)$,
        called as \emph{the termination probability at the state} $x$, characterizes the probability of no event occurring at $x$.
        In particular, if  $\sum_{\sigma \in \Sigma} \rho(x,\sigma) = 1$  holds for any $x \in X$,
        then the system $G$ is called as \emph{nonterminating PDES},
        otherwise, the system $G$ is called as \emph{terminating PDES}.
    \end{itemize}

    \par
    In general, the states that are not reachable from the initial state are meaningless.
    Hence, these states and all the transitions attached to them can be removed.
    This operation is denoted by $Ac(\cdot)$.
    In this paper, we assume that any a automaton (or probabilistic automaton) $G$ is accessible, that is, $G = Ac(G)$.

    \par
    The \emph{logic part of the PDES} $G$, denoted by $logic(G)$, is obtained by dropping the probabilistic module from $G$.
    That is, $logic(G) = \{ X, x_{0}, \Sigma, \delta \}$,
    which is a (conventional) DES.
    \end{definition}

\par
Lawford and Wonham \cite{Lawford} presented a method to extend a terminating system to a nonterminating one.
    However, even if the plant is a nonterminating PDES, it would be changed into a general one during the supervisory control \cite{Kumar}.
    Hence, we would  consider the general PDESs in this paper.

\par
The combination of two PDESs could be characterized by the \emph{product} operation, which is defined as follows.
\par
\begin{definition}
  Given a pair of  PDESs $H_{i} = \{ Q_{i}, q_{0,i}, \Sigma,  \delta_{H_{i}}, \rho_{H_{i}}  \}$,
   $ i \in \{1,2\} $.
  The product of $H_{i}$, denoted by $H_{1} \times H_{2}$, is defined as follows.
    \begin{equation}
      H = H_{1} \times H_{2} = \{ Q_{1} \times Q_{2}, (q_{0,1}, q_{0,2}), \Sigma, \delta_{H}, \rho_{H}  \},
    \end{equation}
    where the $\rho_{H}$ and $\delta_{H}$ are, respectively, defined as
           \begin{equation}
                \rho_{H}((q_1,q_2),\sigma) = \min \{\rho_{H_{1}}(q_{1},\sigma),  \rho_{H_{2}}(q_{2},\sigma) \},
           \end{equation}
      and if  $\rho_{H}((q_1,q_2),\sigma)  > 0$,
          \begin{equation}
            \delta_{H}((q_1,q_2),\sigma) = (\delta_{H_{1}}(q_{1},\sigma),  \delta_{H_{2}}(q_{2},\sigma) ).
          \end{equation}
\end{definition}
\par
According to the definition, it is clear that $logic(H_{1} \times H_{2}) = logic(H_{1}) \times logic(H_{2})$.

\par
The behaviors of a PDES is characterized by the its \emph{generated language}, which is defined as follows.
\begin{definition}
    The probabilistic language generated by the PDES $G$
    is defined as the following mapping  $L_{G}: \Sigma^{*} \rightarrow [0,1]$.
    \begin{equation}
         L_{G}(\epsilon)= 1 ,
    \end{equation}
     where $\epsilon$ is the empty character,
     and for $\forall s \in \Sigma^{*}$ and $ \forall \sigma \in \Sigma$,
     \begin{align}
          L_{G}(s\sigma)  =
          \begin{cases}
           L_G(s) * \rho(\delta(x_0,s),\sigma),           & \text{ if        }  \delta(x_0,s)!,\\
           0,   & \text{ otherwise}. \\
          \end{cases}
  \end{align}

\par
Intuitively, $L_G(s)$ could be viewed as the probability that the string $s$ can be executed in plant $G$.
$L_{G}(\epsilon)= 1$ represents  that a system can always execute the empty character \cite{Kumar}.
\end{definition}

    \par
   Since $\sum_{\sigma \in \Sigma} \rho(x,\sigma) \leq 1 $, it is easy to obtain
    \begin{equation}
      \sum_{\sigma \in \Sigma} L_{G}(s\sigma)\leq L_{G}(s), \forall s \in \Sigma^{*}.
    \end{equation}

\par
 Equations (5) and (7) are exactly the conditions P1) and P2), respectively, in \cite{Kumar}.
 Therefore, the generated language of a PDES is  the \emph{probabilistic language} defined in \cite{Kumar}.
\par

\begin{definition}
  Given a pair of PDESs $H_{1}$ and $H_{2}$. $H_{1}$ and $H_{2}$ are said to be \emph{language-equivalent} if they generate the same probabilistic language,
  that is, $L_{H_{1}}(s) = L_{H_{2}}(s)$, for $\forall s \in \Sigma^{*}$.
\end{definition}

    \par
   The \emph{support language} of a probabilistic language $L$ is defined as $supp(L) =  \{s \in \Sigma^{*} | L(s) > 0\}$  \cite{Kumar}.
It should be pointed out that the support language of a probabilistic language
is always \emph{prefix-closed}.

\par
A  probabilistic language $L$ is \emph{regular} if there exists a finite state probabilistic automaton generating $L$.
Only regular probabilistic languages are considered in this paper for convenience.

\par
In this paper, the behaviors of PDESs are assumed to be partially observed by supervisors.
The partial observation can be characterized by the projection function $P:\Sigma \rightarrow \Sigma_{o}$, which is defined as follows.
  \begin{align}
      P(\sigma) =
      \begin{cases}
       \sigma,           & \text{ if        }  \sigma \in \Sigma_{o},\\
       \epsilon,   & \text{ otherwise}, \\
      \end{cases}
  \end{align}
  where $\epsilon$ is the empty character. It can be extended to $\Sigma^{*}$  by $P(\epsilon) = \epsilon$, and $P(s\sigma)=P(s)P(\sigma)$ for $s \in \Sigma^{*}$ and $\sigma \in \Sigma$.

\par
The common part of two systems behaviors could be characterized by the \emph{intersection} operation defined as follows.
\begin{definition}
  Given two probabilistic languages $L_{1}$ and $L_{2}$ over events set $\Sigma$.
  The \emph{intersection} of $L_{1}$ and $L_{2}$, denoted by $L_{1} \cap L_{2}$, is defined as follows.
         \begin{equation}
                (L_{1} \cap L_{2}) (\epsilon)  = 1;
         \end{equation}
and for $\forall s \in \Sigma^{*}, \forall \sigma \in \Sigma$,
  \begin{align}
      (L_{1} \cap L_{2})(s\sigma) =
      \begin{cases}
     (L_{1} \cap L_{2})(s) &* \min \{ \frac{L_{1}(s\sigma)}{L_{1}(s)}, \frac{L_{2}(s\sigma)}{L_{2}(s)} \}, \\
              & \text{ if    }  s \in supp(L_{1} \cap L_{2}),\\
         0,   & \text{ otherwise}.
      \end{cases}
  \end{align}
\end{definition}
\par
According to the definitions of \emph{product} ($\times$) and \emph{intersection} ($\cap$), it is obvious that
 $ L_{H_{1} \times H_{2}} = L_{H_{1}} \cap L_{H_{2}}$.
  \begin{definition}
    Given two probabilistic languages $L_{1}$ and $L_{2}$ over events set $\Sigma$, $supp(L_{1}) \subseteq supp(L_{2})$.
    $L_{1}$ is said to be a \emph{probabilistic sublanguage} of $L_{2}$, denoted as $ L_{1} \subseteq L_{2}$,
    if for $\forall s \in supp(L_{1})$ and $\forall \sigma \in \Sigma$,
        \begin{equation}
            \frac{L_{1}(s\sigma)}{L_{1}(s)} \leq \frac{L_{2}(s\sigma)}{L_{2}(s)}.
        \end{equation}
\end{definition}
\par

Since $L_{1}(\epsilon) = L_{2}(\epsilon) = 1$, by  induction on the length of the string, it is easy to prove that
$L_{1}(s) \leq L_{2}(s)$, if $ L_{1} \subseteq L_{2} $.
Moreover, it is obvious $L_{1} \subseteq L_{2} \Rightarrow supp(L_{1}) \subseteq supp(L_{2}) $.
\par
In the rest of this paper, the \emph{probabilistic sublanguage} is abbreviated to \emph{sublanguage}.
\par
By means of the definitions of intersection ($\cap$) and sublanguage ($\subseteq$),
the following proposition could be obtained immediately.

\begin{proposition}
   Given two probabilistic languages $L_{1}$ and $L_{2}$ over events set $\Sigma$.
   $L_{1} \cap L_{2} \subseteq L_{1}$; $L_{1} \cap L_{2} \subseteq L_{2}$.
\end{proposition}

\begin{definition}
  Given a pair of PDESs $H_{i} = \{ Q_{i}, q_{0,i}, \Sigma,  \delta_{H_{i}},\rho_{H_{i}}  \}$,
  $ i \in \{1,2\} $.
  $H_{1}$ is called as a \emph{probabilistic subautomaton }of $H_{2}$, denoted as $H_{1} \sqsubseteq H_{2}$,
  if   $Q_{1} \subseteq Q_{2}$, and $q_{0,1} = q_{0,2}$,
  and  for  $\forall q \in Q_{1}$, $\forall \sigma \in \Sigma$,
      \begin{equation}
      [\delta_{H_{1}}(q, \sigma) = \delta_{H_{2}}(q, \sigma)]    \wedge      [\rho_{H_{1}}(q,\sigma) \leq \rho_{H_{2}}(q,\sigma)].
      \end{equation}
\end{definition}

\par
Intuitively, the notion of \emph{probabilistic subautomaton} particularly requires that the state transition diagram of $H_1$ must be a subgraph of that of $H_2$,
and the probabilities of the corresponding transitions in $H_{1}$ must be not larger than that in $H_{2}$.

\par
In the rest of the paper, the \emph{probabilistic subautomaton} is abbreviated to \emph{subautomaton}.

\par
According to the definitions of sublanguage ($\subseteq$) and subautomaton ($\sqsubseteq$),
we have the following propositions.

\par
\begin{proposition}
    $ H \sqsubseteq G \Rightarrow L_{H} \subseteq L_{G} $;
\end{proposition}

\section{partial-observation probabilistic supervisor}
In this section, we define a partial-observation probabilistic supervisor, called as the probabilistic P-supervisor,
and then demonstrate the equivalence of the probabilistic P-supervisor and the scaling-factor function.

\par

The supervisor was defined under the assumption of \emph{full observation} in [16]-[22].
In order to characterize the partial observation of the probabilistic supervisor, we would consider a \emph{partial-observation probabilistic supervisor},
which is defined as follows.
\begin{definition}
  Given plant $ G = \{ X, x_{0}, \Sigma, \delta, \rho  \}$ with $\sigma_{i} \in \Sigma_{c}$ for $i \in [1,m]$, $\sigma_{i} \in \Sigma_{uc}$ for $i \in [m+1,n]$,
  and the observable events set $\Sigma_{o}$.
  $\Theta = \{\theta_{j} | \theta_{j} = \psi \cup \Sigma_{uc}, \psi \in 2^{\Sigma_{c}} \}$ is the set of control patterns.
  The probabilistic P-supervisor $S_{p}: P(L_G) \rightarrow [0,1]^{\Theta}$ is defined as follows: $\forall s \in supp(L_G)$, $P(s) = t$,
\begin{equation}
       S_{p}(t) = \left(
                                                \begin{array}{c}
                                                   p_{0}^{t}   \\
                                                   p_{1}^{t}   \\
                                                   ...      \\
                                                   p_{2^{m}-1}^{t}
                                                \end{array}
                                             \right),
\end{equation}
 where $p_{j}^{t}, j \in [0,2^{m}-1]$, is the probability of control pattern $\theta_{j} \in \Theta$ being adopted by the supervisor when observing $t$.
 These $p_{j}^{t}$ form a probabilistic distribution. That is, $\sum_{j=0}^{2^{m}-1} p_{j}^{t} = 1$.
\end{definition}

\begin{remark}
  Intuitively, $S_{p}(t)$ could be viewed as a special roulette for the observation $t$.
  The roulette issues $2^{m}$ outcomes with the probabilities $p_{j}^{t}, j \in [0,2^{m}-1]$.
  While observing $t$, the supervisor will ``roll" the roulette corresponding to $t$,
  and then adopt the $j$th control pattern if the $j$th outcome is issued.
\end{remark}

 \par

  We could encode any a control pattern $\theta_{j}$, $j \in [0,2^m - 1]$, to an $m$-bits binary number $(b_{m}b_{m-1}....b_{2}b_{1})_{2}$, $b_{i} \in \{0, 1\}$ as follows.
  Firstly, if the controllable event $\sigma_{i} \in \theta_{j}$, then the $i$th bit of the binary number $b_{i} = 1$, otherwise $b_{i} = 0$, $i \in [1,m]$.
  Secondly, let $j = (b_{m}b_{m-1}....b_{2}b_{1})_{2}$.
  Consequently, we obtain such a simple encoding rule.

 \par
 More formally, the containment relationship between the  $i$th controllable event and the $j$th control pattern can  be represented by
    the containment matrix $[IN(i,j)]_{i \in [1,m]}^{j \in [0, 2^m - 1]}$,
    in which  $IN(i,j) = 1$ if $\sigma_{i} \in \theta_{j}$, otherwise $IN(i,j) = 0$.
 Actually, the containment matrix $IN$ is only related to the variable $m$,
 and the $j$th column of matrix $IN$ is the binary form of the decimal integer $j$.

 \par
 Note that the containment matrix $IN$ does not consider the $r$th ($r \in [m+1, n]$) event that is uncontrollable event.
  Since the uncontrollable events are always contained in any control patterns, we can obtain the complete
  containment matrix $\overline{IN}$ by adding $(n-m)$ rows with all the entries being $1$ to the matrix $IN$.

\begin{example}
  Suppose $\Sigma = \{ \sigma_{1}, \sigma_{2}, \sigma_{3} \}$, where $\sigma_{1}, \sigma_{2} \in \Sigma_{c}$ and $\sigma_{3} \in \Sigma_{uc}$.
  According to the encoding rule, since $0 = (00)_2$, $1 = (01)_2$, $2 = (10)_2$, and $3 = (11)_2$,
       the matrices $IN$ and $\overline{IN}$ are as follows.
  \[
        IN  =
        \bordermatrix{%
                    & \theta_{0} & \theta_{1} & \theta_{2} & \theta_{3}  \cr
                    \sigma_{1}   & 0          & 1          & 0     & 1  \cr
                    \sigma_{2}   & 0          & 0          & 1     & 1
        }; \qquad
       \overline{IN}  =
        \bordermatrix{%
                    & \theta_{0} & \theta_{1} & \theta_{2} & \theta_{3}  \cr
                    \sigma_{1}   & 0          & 1          & 0     & 1 \cr
                    \sigma_{2}   & 0          & 0          & 1     & 1 \cr
                    \sigma_{3}   & 1          & 1          & 1     & 1
        }.
  \]
\par
Hence, we obtain the set of control patterns $\Theta = \{ \theta_{0}, \theta_{1}, \theta_{2}, \theta_{3}\}$, where
  $\theta_{0} = \{\sigma_3\}$,
  $\theta_{1} = \{\sigma_1, \sigma_3\}$,
  $\theta_{2} = \{\sigma_2, \sigma_3\}$,
  and $\theta_{3} = \{\sigma_1, \sigma_2, \sigma_3\}$.

\end{example}

\par
Before considering how the probabilistic P-supervisor acts on a PDES, we would discuss how the deterministic P-supervisor acts on a PDES.
Let us see a real-world example first.

\begin{example}

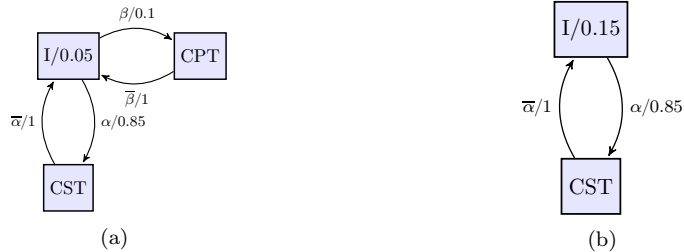
\begin{figure}[htbp]
\centering
\begin{tabular}{C{.5\textwidth}C{.5\textwidth}}
\subfigure [ ] {
    \resizebox{0.25\textwidth}{!}{%
                \centering
            \begin{tikzpicture}[scale=0.25,->,>=stealth',shorten >=1pt,auto,node distance=2.5cm, semithick,
                              every state/.style={fill=blue!10,thick,rectangle}]

                            \node[state] (x0) {I/0.05};
                            \node[state] [below   of = x0] (x1) {CST};
                            \node[state] [right   of = x0] (x2) {CPT};

                            \draw[every node/.style={font=\footnotesize}]
                                   (x0) edge [bend left] node{$\alpha/0.85$} (x1)
                                   (x1) edge [bend left] node{$\overline{\alpha}/1$} (x0)
                                   (x0) edge [bend left] node{$\beta/0.1$}   (x2)
                                   (x2) edge [bend left] node{$\overline{\beta}/1$} (x0);

                            \end{tikzpicture}
    }
} &
\subfigure [ ] {
 \resizebox{0.2\textwidth}{!}{%
           \centering
                                          \begin{tikzpicture}[scale=0.25,->,>=stealth',shorten >=1pt,auto,node distance=2.5cm, semithick,
                              every state/.style={fill=blue!10,thick,rectangle}]

                            \node[state] (x0) {I/0.15};
                            \node[state] [below   of = x0] (x1) {CST};

                            \draw[every node/.style={font=\footnotesize}]
                                   (x0) edge [bend left] node{$\alpha/0.85$} (x1)
                                   (x1) edge [bend left] node{$\overline{\alpha}/1$} (x0) ;
                            \end{tikzpicture}
    }
}
\end{tabular}
\caption{(a). The plant to be controlled; (b). The controlled plant. }
\end{figure}

  Suppose a group of customer-service staffs in a call center are responsible for answering the phones about  consultations and complaints.
  Assume that all the phones could be answered timely, but the group's workload at rush hours (the plant to be controlled) is quite heavy, which is shown in Fig. 1 - (a).
  The state ``I/0.05" denotes the customer-service staffs being available with the probability 0.05;
  and the event $\alpha/0.85$  denotes  the consultation  phones coming in,
   and the plant turning to the ``CST" state (answering the  consultation phone) with the probability $0.85$;
   and the event $\beta/0.1$ denotes  the  complaint phones coming in,
   and the plant turning to the ``CPT" state (answering the  complaint phone) with the probability $0.1$.

  \par
  To reduce the workload of the group at rush hours, an adjusting plan (the control decision) is made by the manager (the supervisor)
  such that the complaint phones are not allowed to come in at rush hours (the event $\beta$ is not allowed to occur at the state ``I").
  It is not hard to understand that such a decision would decrease the probability of the event $\beta$ to zero,
   and then increase the probability of the group being available (the termination probability at state ``I") accordingly.
   In addition, since all the consultation phones could be answered timely at first,
   hence, the decision would not affect the probability of the consultation phones coming in (the probability of the event $\alpha$).
  Therefore, we obtain the group's workload after the adjusted plan is adopted (the controlled plant), which is shown in Fig. 1 - (b).
\end{example}

\par
\begin{remark}
The aforementioned example implies the fact that
for a deterministic control decision $\theta$, if $\sigma \in \theta$, then the probability of the transition with $\sigma$ will remain the same; otherwise, it will decrease to zero.
Formally, we have
\[
    \xi(\sigma) = pr(\sigma) * pr(\sigma \text{ is enabled by } \theta),
\]
where $\xi(\sigma)$ and $pr(\sigma)$ denote the probabilities of the transition with the event $\sigma$ at the current state in the controlled plant and uncontrolled plant, respectively.
Because $\theta$ is a deterministic control decision,
we have $pr(\sigma \text{ is enabled by } \theta) = 1$, if $\sigma \in \theta$; otherwise, $pr(\sigma \text{ is enabled by } \theta) = 0$.
\end{remark}

\par
In what follows, we would generalize the case of the deterministic P-supervisor to the case of probabilistic P-supervisor.
First, we  present the notion of the \emph{controlled transition-probability function},
which characterizes how  a probabilistic P-supervisor acts on the probabilities of the transitions of the plant.
\begin{definition}
  Given a PDES $G = \{ X, x_{0}, \Sigma, \delta, \rho  \}$ with $\sigma_{i} \in \Sigma_{c}$ for $i \in [1,m]$, $\sigma_{i} \in \Sigma_{uc}$ for $i \in [m+1,n]$,
  and the observable events set $\Sigma_{o}$,
    and the  probabilistic P-supervisor: $\forall s \in supp(L_{G})$, $P(s) = t$,
   $S_{p}(t) = [p_{0}^{t} \quad  p_{1}^{t} \quad ... \quad p_{2^{m}-1}^{t}]^{T}$.
     Suppose that $\delta(x_{0}, s) = x$,
     the controlled  transition-probability function by $S_{p}$ $\xi : L_G \times \Sigma \rightarrow [0,1]$ is defined as follows:
     for $\sigma_{i} \in \Sigma$,
 \begin{align}
  \xi(s , \sigma_{i})= &  pr(\sigma_{i} \text{ is active at the state  $x$}) *
                        pr(\sigma_{i} \text{ is enable by $S_{p}(t)$}) \nonumber \\
                     = & \rho(x, \sigma_{i}) * \sum_{j=0 }^{2^m -1}  pr(\theta_{j} \text{ is adopted}) *
                         pr(\sigma_{i} \text{ is enable } | \theta_{j} \text{ is adopted})   \nonumber\\
                     = & \rho(x, \sigma_{i}) * \sum_{j=0 }^{2^m -1} \{ p_{j}^{t} * \overline{IN}(i,j)\}.
 \end{align}

\par
     Note that $\forall i \in [1,m]$, that is, $\sigma_i \in \Sigma_{c}$, and $\forall j\in [0, 2^m - 1]$, $\overline{IN}(i,j) \leq 1$.
     Hence,
            \begin{align}
                             \xi(s , \sigma_{i}) &\leq \rho(x, \sigma_{i}) * \sum_{j=0 }^{2^m -1}  p_{j}^{t}
                                                 = \rho(x, \sigma_{i}),  \text{ if $i \in [1, m]$ }.
             \end{align}

 \par
     On the other hand, if $i \in [m+1, n]$, that is, $\sigma_i \in \Sigma_{uc}$, and $\forall j\in [0, 2^m - 1]$, $\overline{IN}(i,j) = 1$.
     Hence,
            \begin{align}
                             \xi(s , \sigma_{i}) &= \rho(x, \sigma_{i}) * \sum_{j=0 }^{2^m -1}  p_{j}^{t}
                                                            = \rho(x, \sigma_{i}), \text{ if $i \in [m+1, n]$ }.
             \end{align}

 \end{definition}

 \begin{remark}
   Actually, $\xi(s , \sigma_{i})$ is the probability of the transition  with event $\sigma_{i}$ at the state reaching by sequence $s$ in the controlled system $S_{p}/G$.
    Equations (15) and (16) have the following meaning:
    under the probabilistic supervisory control, the probability of the transitions with the controllable events will decrease usually;
    however, the probability of the transitions with the uncontrollable events will remain the same.
   It meets the peoples' intuition about the supervisory control theory.
   That is, the supervisory control always limits the behaviors of the plant, unless the behaviors are uncontrollable.
 \end{remark}

 By means of the above notion, we could present the \emph{controlled probabilistic language}  $L_{S_{p}/G}$,
 which denotes the behaviors of the controlled plant $S_{p}/G$.
 \begin{definition}
   Given the PDES $G = \{ X, x_{0}, \Sigma, \delta, \rho  \}$ with $\sigma_{i} \in \Sigma_{c}$ for $i \in [1,m]$, $\sigma_{i} \in \Sigma_{uc}$ for $i \in [m+1,n]$, and the observable events set $\Sigma_{o}$,
    and the probabilistic P-supervisor $S_{p}$.
     The controlled probabilistic language  $L_{S_{p}/G}$ is defined by the following recursive manner:
         \begin{equation}
                L_{S_{p}/G}(\epsilon) = 1,
         \end{equation}
and $\forall s \in \Sigma^{*}, \forall \sigma_{i} \in \Sigma$, $i \in [1,n]$,
        \begin{equation}
             L_{S_{p}/G}(s\sigma_{i})  =  L_{S_{p}/G}(s) * \xi(s , \sigma_{i}). 
         \end{equation}
 \end{definition}

\par
    Note that the controlled probabilistic language $L_{S_{p}/G}$ is only determined by the controlled transition-probability function $\xi$ defined in Equation (14), which can be computed by using the probabilistic P-supervisor $S_{p}$ and the plant $G$. Obviously,
    the controlled probabilistic language is usually a terminating one, even the uncontrolled plant is a nonterminating system.

\par
     For each observation $t \in P(L_G)$, the probabilistic P-supervisor has defined a  vector
     that has a quite big size dimensions ($2^{m}$).
     At the end of this section, we would consider how to construct  a more compact form for the probabilistic P-supervisor.

\par
\begin{definition}
   Given the PDES $G = \{ X, x_{0}, \Sigma, \delta, \rho  \}$ with $\sigma_{i} \in \Sigma_{c}$ for $i \in [1,m]$, $\sigma_{i} \in \Sigma_{uc}$ for $i \in [m+1,n]$, and the observable events set $\Sigma_{o}$.
   The scaling-factor function $K: P(L_G) \rightarrow [0,1]^{\Sigma} $ is defined as follows:  $\forall s \in supp(L_{G})$ such that
   $\delta(x_{0}, s) = x$,   $P(s) = t$,
        \begin{equation}
               K(t) = \left(
                                                        \begin{array}{c}
                                                           k_{1}^{t}   \\
                                                           k_{2}^{t}   \\
                                                           ...      \\
                                                           k_{n}^{t}
                                                        \end{array}
                                                     \right),
        \end{equation}
        where
 $ 0 \leq k_{i}^{t} \leq 1$ for $i \in [1, m]$, and $k_{i}^{t} = 1$ for $i \in [m+1, n]$.
 $ K(t)(\sigma_i) = k_{i}^{t}$, $i \in [1, n]$, is called as the \emph{scaling-factor of the transition with $\sigma_{i}$ at $x$}.

\end{definition}

\par
The following notion characterizes the controlled behaviors of the plant by the scaling-factor function $K$.

\par
 \begin{definition}
   Given the PDES $G = \{ X, x_{0}, \Sigma, \delta, \rho  \}$ with $\sigma_{i} \in \Sigma_{c}$ for $i \in [1,m]$, $\sigma_{i} \in \Sigma_{uc}$ for $i \in [m+1,n]$, and the observable events set $\Sigma_{o}$,
    and the  scaling-factor function $K: P(L_G) \rightarrow [0,1]^{\Sigma} $.
    The controlled probabilistic language  $L_{K/G}$ by $K$ is defined by the following recursive manner:
         \begin{equation}
                L_{K/G}(\epsilon) = 1,
         \end{equation}
and  $\forall s \in \Sigma^{*}$, $P(s) = t$, such that $\delta(x_0, s) = x$, and $\forall \sigma_{i} \in \Sigma$,
        \begin{equation}
             L_{K/G}(s\sigma_{i})  =  L_{K/G}(s) *\rho(x, \sigma_{i}) * K(t)(\sigma_{i}).
         \end{equation}
 \end{definition}

\par
By means of Definition 8 and Definition 10, if
\begin{equation}
          K(t)(\sigma_i) = \sum_{j=0 }^{2^m -1} \{ p_{j}^{t} * \overline{IN}(i,j)\}, i \in [1,n],
\end{equation}
then $L_{S_{p}/G} = L_{K/G}$, that is, the controls by  the probabilistic P-supervisor $S_{p}$ and the scaling-factor function  $K$ are equivalent to each other.

\par
In what follows, we would prove this equivalence by proving the solvability of Equation (22).
The celebrated \emph{Farkas's Lemma} in linear algebra is necessary for the proof.

\par
\begin{lemma}
(\textbf{Farkas' Lemma \cite{Farkas}}) Let $\mathbf{A} \in \mathbb{R}^{p \times q}$ and $\mathbf{b} \in \mathbb{R}^{p}$. Then exactly one of the following two statements is true:
        \begin{enumerate}
                \item
                        There exists an $\mathbf {x} \in \mathbb {R} ^{q}$, such that $\mathbf {A \times x} =\mathbf {b}$ and $\mathbf {x} \geq 0$.
                \item
                        There exists a $\mathbf {y} \in \mathbb {R} ^{p}$, such that $\mathbf {A} ^{\mathsf {T}} \times \mathbf {y} \geq 0$ and $\mathbf {b} ^{\mathsf {T}} \times \mathbf {y} <0$.
        \end{enumerate}
\end{lemma}
Here the notation $ \mathbf {x} \geq 0 $  means that all components of the vector $ \mathbf {x} $ are nonnegative.

\begin{theorem}
   Given a PDES $G = \{ X, x_{0}, \Sigma, \delta, \rho  \}$ with $\sigma_{i} \in \Sigma_{c}$ for $i \in [1,m]$, $\sigma_{i} \in \Sigma_{uc}$ for $i \in [m+1,n]$, and the observable events set $\Sigma_{o}$,
   and  $L \subseteq L_{G}$.
  There exists a probabilistic P-supervisor $S_{p}$ synthesizing the controlled probabilistic language $L$,
   if and only if there exists a scaling-factor function $K$ synthesizing the controlled probabilistic language $L$.
\end{theorem}
\begin{proof}
  For necessity,
  it is sufficient to show that given a probabilistic distribution $p_{j}^{t}, j \in [0, 2^m -1]$,
  the Equation (22) has  solutions for $ K(t)(\sigma_{i})$, $i \in [1,n]$, satisfying $ 0 \leq  K(t)(\sigma_{i}) \leq 1$ for $i \in [1, m]$, and $K(t)(\sigma_{i}) = 1$ for $i \in [m+1, n]$.
   It is obvious.

  \par
  For sufficiency, it is sufficient to show that given $K(t)(\sigma_{i})$, $i \in [1,n]$, satisfying $ 0 \leq  K(t)(\sigma_{i}) \leq 1$ for $i \in [1, m]$, and $K(t)(\sigma_{i}) = 1$ for $i \in [m+1, n]$,
  the Equation (22) has nonnegative solutions  for $p_{j}^{t}, j \in [0, 2^m -1]$.

   \par
   We construct an $(m + 1) \times 2^{m}$  matrix $\widehat{IN}$ by choosing the first $(m+1)$ rows of the matrix $\overline{IN}$.
   Then the Equation (22) has nonnegative solutions for $p_{j}^{t}, j \in [0, 2^m -1]$,
   if and only if given $K(t)(\sigma_{i})$, $i \in [1,m + 1]$, satisfying $ 0 \leq  K(t)(\sigma_{i}) \leq 1$ for $i \in [1, m]$ and $K(t)(\sigma_{m + 1}) = 1$,
   \begin{equation}
          K(t)(\sigma_i) = \sum_{j=0 }^{2^m -1} \{ p_{j}^{t} * \widehat{IN}(i,j)\}, i \in [1,m+1],
    \end{equation}
  has nonnegative solutions for $p_{j}^{t}, j \in [0, 2^m -1]$.

  \par
  Let $\widehat{K}(t) = ( K(t)(\sigma_1) \quad \ldots \quad  K(t)(\sigma_{m+1}) )  ^{\mathsf {T}} $.
  We could rewrite the Equation (23) as the following matrix form.
  \begin{equation}
                \widehat{IN} \times S_{p} = \widehat{K}(t).
  \end{equation}
     \par
Then by means of Lemma 1 (Farkas's Lemma), to prove Equation (23) has nonnegative solutions for $p_{j}^{t}, j \in [0, 2^m -1]$,
it is sufficient to show that
 \begin{equation}
    [\widehat{IN} ^{\mathsf {T}} \times \mathbf {y} \geq 0 ]  \wedge    [\widehat{K}(t) ^{\mathsf {T}} \times \mathbf {y} < 0],
 \end{equation}
 has no solution for $\mathbf {y} \in \mathbb{R} ^{m+1}$.

 \par
 Actually, $\widehat{IN} ^{\mathsf{T}} = ( IN^{\mathsf {T}} \quad  \textbf{1} )$, where $\textbf{1}$ is a $2^{m}$ rows vector with all components being 1.
 Suppose $\mathbf {y} = (y_{1} \quad \ldots $ $\quad y_{m+1})^{\mathsf{T}}$.
 Then $ \widehat{IN} ^{\mathsf{T}} \times \mathbf{y} \geq 0$ implies that
 \begin{multline}
  IN^{\mathsf {T}} \times   (y_{1}   \quad \ldots \quad y_{m-1} \quad y_{m})^{\mathsf{T}} +
                            (y_{m+1} \quad \ldots \quad y_{m+1} \quad y_{m+1})^{\mathsf{T}}
                                            \geq 0.
 \end{multline}
According to the definition of $IN$,
$IN^{\mathsf {T}} \times   (y_{1} \quad \ldots \quad y_{m-1} $ $ \quad y_{m})^{\mathsf{T}}$ is a $2^{m}$ rows vector,
 in which the components enumerate all the sums of the subsets of $\{ y_{j}\}_{j=1}^{m}$.
 We denote the maximum value of all the sums of the subsets of set $B$ as $MaxSumSub(B)$.
 That is, $MaxSumSub(B) = \max \Big \{ \sum_{ a \in A }  a   | A \subseteq B \Big \}$.
 Then Equation (26) means that
 \begin{equation}
  y_{m+1} \geq MaxSumSub ( \{ -y_{j} \}_{j=1}^{m} ).
  \end{equation}

\par
On the other hand, since $K(t)(\sigma_{m+1}) = 1$, $\widehat{K}(t) ^{\mathsf {T}} \times \mathbf {y} < 0$ means that
\begin{equation}
  y_{m+1} < \sum_{i=1}^{m} \{ K(t)(\sigma_{i}) * (-y_{i}) \}.
\end{equation}
Let $p_{i} = 1$, if $(-y_{i}) \geq 0$; otherwise $p_{i} = 0$, $i \in [1, m]$. Then we have
\begin{align}
  y_{m+1} &<     \sum_{i=1}^{m} \{ K(t)(\sigma_{i}) * (-y_{i}) \} \nonumber \\
          &\leq  \sum_{i=1}^{m} \{ p_{i} * (-y_{i}) \}  \text{ (according to }  0 \leq K(t)(\sigma_{i}) \leq 1) \nonumber \\
          &\leq   MaxSumSub ( \{ -y_{j} \}_{j=1}^{m} ),
\end{align}
 which contradicts Equation (27). Thus, Equation (25) has no solution for $\mathbf {y} \in \mathbb{R} ^{m+1}$.
     Hence, Equation (23) has nonnegative solutions for $p_{j}^{t}, j \in [0, 2^m -1]$.
     As a result, Equation (22) also has nonnegative solutions for $p_{j}^{t}, j \in [0, 2^m -1]$.
    This completes the proof of the sufficiency.
\end{proof}

\begin{remark}
  Theorem 1 demonstrates that the scaling-factor function $K: P(L_G) \rightarrow [0,1]^{\Sigma} $
   is exactly a compact form of the probabilistic P-supervisor $S_{p}: P(L_G) \rightarrow [0,1]^{\Theta}$ actually.
  For the simplicity in the supervisory control of PDESs, we could compute the scaling-factor function $K$ first,
  and then obtain the probabilistic P-supervisor $S_{p}$ by solving the Equation (22) when it is needed.
\end{remark}

\section{Probabilistic Supervisory Control Theory of PDESs}
In this section, we first present the notions of the probabilistic controllability and observability,
 and then show that these notions serve as the necessary and sufficient conditions for the existence of the probabilistic P-supervisors.
 Moreover, we would present two polynomial algorithms to verify the probabilistic controllability and observability.

\subsection{Probabilistic Controllability and Observability Theorem}

\begin{definition}
  Given a plant $G = \{ X, x_{0}, \Sigma, \delta, \rho  \}$ with the controllable events set $\Sigma_{c}$, and the probabilistic specification
  $L$, such that $L \subseteq  L_{G}$. Suppose $L$ is generated by the probabilistic automaton $H = \{ Q, q_{0}, \Sigma, \delta_{H}, \rho_{H}  \}$,
  that is, $L = L_{H}$.
   The specification $L$ and its generator $H$ are said to be probabilistic controllable w.r.t. $G$ and $\Sigma_{c}$,
   if  $\forall s \in supp(L_{H})$ such that $\delta( x_0 , s ) = x $ and $\delta_{H}( q_0 , s ) = q $, and $\forall \sigma \in \Sigma_{uc}$,
        \begin{equation}
           \rho_{H}(q,\sigma) = \rho(x,\sigma).
    \end{equation}
\end{definition}

\par
The notion of the \emph{probabilistic controllability} characterizes an important principle of the supervisory control theory that the supervisory control cannot limit the uncontrollable behaviors of the plant.

\par

\begin{definition}
   Given a plant $G = \{ X, x_{0}, \Sigma, \delta, \rho  \}$
   with the controllable events set $\Sigma_{c}$, the observable events set $\Sigma_{o}$,
   and the probabilistic specification $L$, such that $L \subseteq  L_{G}$.
   Suppose $L$ is generated by the probabilistic automaton $H = \{ Q, q_{0}, \Sigma, \delta_{H}, \rho_{H}  \}$,
    that is, $L = L_{H}$.
   The specification $L$ and its generator $H$
    are said to be probabilistic observable w.r.t. $G$, $\Sigma_{c}$ and $\Sigma_{o}$,
     if $\forall s_{1},s_{2} \in supp(L_{H})$, such that $P(s_{1}) = P(s_{2})$, and
    $\delta( x_0 , s_{i} ) = x_{i} $ and $\delta_{H}( q_0 , s_{i} ) = q_{i} $, $i = \{1 , 2\}$,
    and $\forall \sigma \in \Sigma_{c}$,
    \begin{equation}
             \rho(x_{1},\sigma) * \rho_{H}(q_{2},\sigma)  = \rho(x_{2},\sigma) * \rho_{H}(q_{1},\sigma).
    \end{equation}
\end{definition}

\par
If $\rho(x_{1},\sigma) \neq 0$ and $\rho(x_{2},\sigma) \neq 0$, Equation (31) could be rewritten as $\frac{\rho_{H}(q_{1},\sigma)}{\rho(x_{1},\sigma)} = \frac{\rho_{H}(q_{2},\sigma)}{\rho(x_{2},\sigma)}$,
which  characterizes another principle of the supervisory control theory that
if the supervisor cannot differentiate between two states, then these states should require the same control action.
\par

\begin{remark}
  In order to reflect the intuitive meanings of  the \emph{probabilistic controllability} and \emph{probabilistic observability} more clearly,
  their definitions are presented by the automata form.
  According to the definitions, if the generator $H$ is probabilistic controllable (observable), then any its language-equivalent generator $H^{*}$ is also
  probabilistic controllable (observable).
\end{remark}

 \par
 It should be pointed out that
  the notions of the \emph{probabilistic controllability} and \emph{probabilistic observability} introduced here are the extensions of the notions of the \emph{controllability \cite{first-supevisory}} and \emph{observability \cite{observability},\cite{observability2}}, respectively.
  The relations between these notions will be discussed in Section \uppercase\expandafter{\romannumeral5}.

\par
The following theorem demonstrates that the probabilistic controllability and observability are the necessary and sufficient conditions for
the existence of the probabilistic P-supervisors.

\par

\begin{theorem}(\textbf{Probabilistic Controllability and Observability Theorem})
 Given a plant $G = \{ X, x_{0}, \Sigma, \delta, \rho  \}$ with the controllable events set $\Sigma_{c}$ and
     the observable events set $\Sigma_{o}$,
     and the probabilistic specification $L$, such that $L \subseteq  L_{G}$.
     There exists a probabilistic P-supervisor $S_{p}:P(L_{G}) \rightarrow [0,1]^{\Theta} $, such that  $L_{S_{p}/G} = L$,
     if and only if the specification $L$ is probabilistic controllable w.r.t. $G$ and $\Sigma_{c}$,
      and probabilistic observable w.r.t. $G$, $\Sigma_{c}$ and $\Sigma_{o}$.
\end{theorem}
\begin{proof}
By means of Theorem 1, it is sufficient to prove the following claim.

\par
\emph{
There exists a scaling-factor function $K$, such that  $L_{K/G} = L$,
     if and only if $L$ is probabilistic controllable w.r.t. $G$ and $\Sigma_{c}$,
      and probabilistic observable w.r.t. $G$, $\Sigma_{c}$ and $\Sigma_{o}$.
      }

\par
In what follows, we would prove the correctness this claim.

\par
Suppose $L$ is generated by the probabilistic automaton $H = \{ Q, q_{0}, \Sigma, \delta_{H}, \rho_{H}  \}$, that is, $L = L_{H}$.

\par
For necessity, if there exists a scaling-factor function $K$, such that $L_{K/G} = L_{H}$,
     then we need to prove that the specification $L_{H}$ is probabilistic controllable and observable.
     \par
     First of all, we show that $L_{H}$ is probabilistic controllable.
     For any $s \in supp(L_{H})$, suppose $\delta(x_0,s) = x$ and $\delta_{H}(q_0,s) = q$, and $P(s) = t$. Then
     \begin{align}
       \rho_{H}(q,\sigma) = & L_{H}(s\sigma)/L_{H}(s) \text{  (by the definition of }      \text{ probabilistic languages) }           \nonumber    \\
                          = & L_{K/G}(s\sigma)/L_{K/G}(s)    \text{  (by  $L_{K/G} = L_{H}$) } \nonumber    \\
                          = & \rho(x, \sigma) * K(t)(\sigma) \text{  (by  Equation (21)). }
     \end{align}

     \par
     Since $K(t)(\sigma) = 1$ holds for $\forall \sigma \in \Sigma_{uc}$,
      Thus $\rho_{H}(q,\sigma) = \rho(x, \sigma)$ for  $\forall  \sigma \in \Sigma_{uc}$.
       Therefore, $L_{H}$ is probabilistic controllable.

       \par
       Secondly, we show that $L_{H}$ is probabilistic observable. For $\forall s_{1}, s_{2} \in supp(L_{H})$ and $\forall \sigma \in \Sigma_{c}$,
       suppose $\delta(x_0,s_{i}) = x_{i}$ and $\delta_{H}(q_0,s_{i}) = q_{i}$ and $P(s_{i}) = t$,
       $i = \{1, 2\}$.  Then similar to Equation (32), we obtain
       \begin{subequations}
              \begin{gather}
                          \begin{align}
                                     \rho_{H}(q_{1},\sigma) &= \rho(x_{1}, \sigma) * K(t)(\sigma), \\
                                     \rho_{H}(q_{2},\sigma) &= \rho(x_{2}, \sigma) * K(t)(\sigma).
                          \end{align}
              \end{gather}
      \end{subequations}
      We have
             \begin{subequations}
              \begin{gather}
                          \begin{align}
                                     \rho(x_{2}, \sigma) * \rho_{H}(q_{1},\sigma) &= \rho(x_{2}, \sigma) * \rho(x_{1}, \sigma) * K(t)(\sigma), \\
                                     \rho(x_{1}, \sigma) * \rho_{H}(q_{2},\sigma) &= \rho(x_{1}, \sigma) * \rho(x_{2}, \sigma) * K(t)(\sigma).
                          \end{align}
              \end{gather}
      \end{subequations}
     Hence, $\forall s_{1}, s_{2} \in supp(L_{H})$, $P(s_{1}) = P(s_{2})$, and $\forall \sigma \in \Sigma_{c}$, we have
     \begin{equation}
       \rho(x_{1}, \sigma) * \rho_{H}(q_{2},\sigma) = \rho(x_{2}, \sigma) * \rho_{H}(q_{1},\sigma).
     \end{equation}
     Therefore, $L_{H}$ is probabilistic observable. This completes the proof of the necessity.

     \par
     For sufficiency, we need to show if $L_{H}$ is probabilistic controllable and observable, then there exists  a scaling-factor function $K$ such that $L_{K/G} = L_{H}$.

     \par
     For $\forall t \in P(supp(L_{G}))$ and $\forall \sigma \in \Sigma$, define the scaling-factor function $K(t)$ as follows.
      \begin{align}
          K(t)(\sigma) =
              \begin{cases}
               1,           & \text{ if        }  \sigma \in \Sigma_{uc};\\
               \frac{\rho_{H}( \delta_{H}(q_0,s) ,\sigma)}{\rho( \delta(x_0,s),\sigma)},   & \text{ else if } \exists s \in supp(L_{G}), \text{ such that } \\
               & P(s) = t, \delta(x_0,s\sigma)! \text{ and } \delta_{H}(q_0,s)!; \\
               0,           & \text{ else. }
              \end{cases}
      \end{align}

      \par
      First of all, we need to show the scaling-factor function $K$ defined in Equation (36) is well-defined.
      It is sufficient to show  $\forall s_{i} \in supp(L_{G})$ with $P(s_{i})= t$, $\delta(x_0,s_{i}\sigma)!$  and $\delta_{H}(q_0,s_{i})! $, $i \in \{1,2\}$, and $\forall \sigma \in \Sigma_{c}$,
      \begin{equation}
            \frac{\rho_{H}( \delta_{H}(q_0,s_1) ,\sigma)}{\rho( \delta(x_0,s_1),\sigma)} = \frac{\rho_{H}( \delta_{H}(q_0,s_2) ,\sigma)}{\rho( \delta(x_0,s_2),\sigma)}.
      \end{equation}
      By the probabilistic  observability of $L_{H}$, we immediately obtain Equation (37).

      \par
      Secondly, we would show that with the scaling-factor function $K$ defined in Equation (36), $L_{K/G} = L_{H}$.
      The proof is by induction on the length of the string $s \in \Sigma^{*}$.

      \par
      The base case is for $|s| = 0$.  $L_{K/G}(\epsilon) = L_{H}(\epsilon) = 1$. Hence, the base case holds.

      \par
      Suppose for $|s| \leq n$, $L_{K/G}(s) = L_{H}(s) $ holds. Then we need to show $\forall \sigma \in \Sigma$, $L_{K/G}(s\sigma) = L_{H}(s\sigma)$.
      By Equation (21) and $L_{K/G}(s) = L_{H}(s)$, we have
      \begin{equation}
            L_{K/G}(s\sigma)  = L_{H}(s) *   \rho( \delta(x_{0}, s)  ,\sigma) * K(P(s))(\sigma).
      \end{equation}

      \par
      We prove $L_{K/G}(s\sigma) = L_{H}(s\sigma)$ by dividing into the following three cases.
      \begin{enumerate}
            \item
                If $\sigma \in \Sigma_{uc}$, then $K(P(s))(\sigma) = 1$.
                According to the probabilistic controllability of $L_{H}$, we obtain $\rho( \delta(x_{0}, s)  ,\sigma)  = \rho_{H}( \delta_{H}(q_{0}, s)  ,\sigma)$.
                Thus,
                \begin{align*}
                L_{K/G}(s\sigma) & = L_{H}(s) *  \rho( \delta(x_{0}, s)  ,\sigma) * K(P(s))(\sigma)\\
                                 & = L_{H}(s) * \rho_{H}( \delta_{H}(q_{0}, s)  ,\sigma) * 1 \\
                                 & = L_{H}(s\sigma).
                \end{align*}

            \item
                If $\sigma \in \Sigma_{c}$, and $\delta(x_0,s\sigma)!$  and  $\delta_{H}(q_0,s)!$,
                \begin{align*}
                   L_{K/G}(s\sigma) & = L_{H}(s) *   \rho( \delta(x_{0}, s)  ,\sigma) * K(P(s))(\sigma) \\
                                                & = L_{H}(s) *   \rho( \delta(x_{0}, s)  ,\sigma) * \frac{\rho_{H}( \delta_{H}(q_0,s) ,\sigma)}{\rho( \delta(x_0,s),\sigma)} \\
                                                & = L_{H}(s) * \rho_{H}( \delta_{H}(q_{0}, s)  ,\sigma) \\
                                                & = L_{H}(s\sigma).
                \end{align*}

            \item
                If $\sigma \in \Sigma_{c}$, and $\delta(x_0,s\sigma)$\sout{!} or  $\delta_{H}(q_0,s)$\sout{!},
                then $\delta_{H}(q_0,s\sigma)$\sout{!}.
                Hence, $L_{H}(s\sigma) = 0$.
                On the other hand, $ L_{K/G}(s\sigma) = L_{H}(s) *   \rho( \delta(x_{0}, s)  ,\sigma) * K(P(s))(\sigma)  = 0$.
                 Therefore, $ L_{K/G}(s\sigma) = L_{H}(s\sigma)$.
      \end{enumerate}

      \par
       This completes the proof of the sufficiency.
\end{proof}
\begin{remark}
  Theorem 2 not only demonstrates that the probabilistic controllability and observability are the necessary and
  sufficient conditions for the existence of the probabilistic P-supervisors,
  but also provides the design method of the probabilistic P-supervisor.
  Since Equation (36) formulates a scaling-factor function that can synthesize the desired specification,
  we can obtain the probabilistic P-supervisor by solving Equation (22).
\end{remark}

\par
The following simple example illustrates how to verify the probabilistic controllability and observability by definitions,
and how to use the probabilistic supervisor to achieve the desired probabilistic specification.

\begin{figure}[htp]
  \centering
  \includegraphics[width=0.35\textwidth]{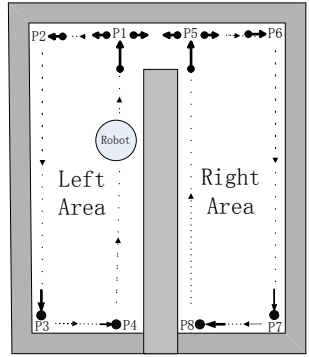}
  \caption{mobile robot with sensors}
  \label{fig1}
\end{figure}

\begin{figure}[htbp]
\centering
\begin{tabular}{C{.4\textwidth}C{.4\textwidth}}
\resizebox{0.28\textwidth}{!}{%
              \begin{tikzpicture}[scale=0.25,->,>=stealth',shorten >=1pt,auto,node distance=3.8cm, semithick,
                              every state/.style={fill=blue!10,thick,font=\LARGE}]

                            \node[state] (x0) {$x_0$};
                            \node[state] [above    of = x0] (x1) {$x_1$};
                            \node[state] [below left   of = x0] (x2) {$x_2$};
                            \node[state] [below right  of = x0] (x3) {$x_3$};

                            \draw[every node/.style={font=\large}][align=center]
                                   (x0) edge  [bend left] node{$\sigma_{3}/0.25$} (x1)
                                   (x0) edge   node{$\sigma_{4}/0.375$} (x2)
                                   (x0) edge  [bend left] node{$\sigma_{5}/0.375$} (x3)
                                   (x1) edge  [bend left] node {$\sigma_{1}/0.5; $ \\ $ \sigma_{2}/0.5$} (x0)
                                   (x2) edge  [bend left] node {$\sigma_{2}/1 $} (x0)
                                   (x3) edge   node {$\sigma_{1}/1 $} (x0);
                            \end{tikzpicture}
}
&
\resizebox{0.28\textwidth}{!}{%
              \begin{tikzpicture}[scale=0.25,->,>=stealth',shorten >=1pt,auto,node distance=3.8cm, semithick,
                              every state/.style={fill=blue!10,thick,font=\LARGE}]

                            \node[state] (q0) {$q_0$};
                            \node[state] [above    of = x0] (q1) {$q_1$};
                            \node[state] [below left   of = q0] (q2) {$q_2$};
                            \node[state] [below right  of = q0] (q3) {$q_3$};

                            \draw[every node/.style={font=\large}][align=center]
                                   (q0) edge  [bend left] node{$\sigma_{3}/0.25$} (q1)
                                   (q0) edge   node{$\sigma_{4}/0.375$} (q2)
                                   (q0) edge  [bend left] node{$\sigma_{5}/0.375$} (q3)
                                   (q1) edge  [bend left] node {$\sigma_{1}/0.4; $ \\ $ \sigma_{2}/0.5$} (q0)
                                   (q2) edge  [bend left] node {$\sigma_{2}/1 $} (q0)
                                   (q3) edge   node {$\sigma_{1}/1 $} (q0);
                            \end{tikzpicture}
}
\end{tabular}
\caption{(a). $G$ in Example 3; (b). $H$ in Example 3. }
\end{figure}
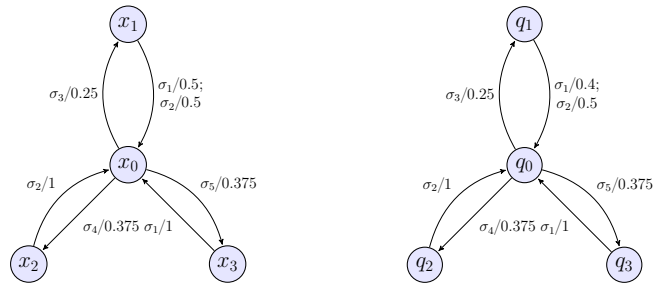

\par
\begin{example}
The example is obtained by slightly altering the example in \cite{Swarm}.
Consider a robot equipped with sensors that can detect obstacles.
Assume the robot starts moving upwards in the left area, as shown in Fig. 2.
We model the robot as the PDES  $G = \{ X, x_{0}, \Sigma, \delta, \rho  \}$,   $\Sigma = \{ \sigma_{i} \}$, $i \in [1,5]$, as shown in Fig. 3 - (a).
    Suppose $\Sigma_{c} = \Sigma_{o} =  \{\sigma_{1}, \sigma_{2} \}$.
   The uncontrollable events $\sigma_{3}$, $\sigma_{4}$ and $\sigma_{5}$ denote the sensing of an obstacle ``in front",``in front and on the right" and ``in front and on the left", respectively.
   The controllable events $\sigma_{1}$ and $\sigma_{2}$ represent that the robot choose the actions of ``turning right" and ``turning left", respectively.
   The initial state $x_{0}$ means the robot moves keeping the original direction.
   The state $x_{1}$ means the robot might be at the positions ``P1" or ``P5", in which the alternative actions of ``turning right" (denoted by event ``$\sigma_{1}$")
    and ``turning left" (denoted by event ``$\sigma_{2}$") are randomly chosen.
   In addition, the state $x_{2}$ means the robot might be at the positions ``P2", ``P3" or ``P4", in which the only action can be chosen is  ``turning left".
   Moreover, the state $x_{3}$ means the robot might be at the positions ``P6", ``P7" or ``P8", in which the only action can be chosen is  ``turning right".

\par
Obviously, the uncontrolled robot explores the left area and right area with the same probability.
Suppose the desired ratio of the probabilities of searching the left area and right area is $\frac{5}{4}$.
Then the probabilistic specification could be characterized by $H = \{ Q, q_{0}, \Sigma, \delta_{H}, \rho_{H}  \}$, as shown in Fig. 3 - (b).

\par
The deterministic supervisors cannot realize such a probabilistic specification, as the control pattern $\{\sigma_{1},\sigma_{2},\sigma_{3},\sigma_{4},\sigma_{5}\}$ at state $x_{1}$
will do not change the original probabilities;
and the control patterns $\{ \sigma_{1}, \sigma_{3},\sigma_{4},\sigma_{5} \}$ and $\{\sigma_{2}, \sigma_{3},\sigma_{4},\sigma_{5} \}$ will make the robot only explore at the right area and left area, respectively;
and the control patterns $\{ \sigma_{3},\sigma_{4},\sigma_{5} \}$ will block the plant at state $x_{1}$.
  \par

  However, the probabilistic P-supervisors can be competent this control task, as the specification is probabilistic and only partial events can be observed.

\par
  We discuss the  probabilistic controllability and observability of the specification $H$.
  Firstly, the probabilistic controllability of $H$ obviously holds, as the probabilities of the corresponding transitions with uncontrollable events in $G$ and $H$  are all equal to each other.
  We continue to investigate the probabilistic observability of $H$.
  Note that there exist the events sequences
  $s_{1} = \sigma_{3} $ and $s_{2} = \sigma_{5} $, such that $P(s_{1}) = P(s_{2})$, reaching the states  $q_{1}$ and $q_{3}$, respectively.
  Note that $0.4 = \rho_{H}(q_{1},\sigma_{1}) * \rho(x_{3},\sigma_{1}) \neq \rho_{H}(q_{3}, \sigma_{1}) * \rho(x_{1}, \sigma_{1}) = 0.5$.
  Thus, $H$ is not probabilistic observable.
\par

Suppose the observable events set is revised to $\Sigma_{o} =  \{\sigma_{1}, \sigma_{2}, \sigma_{3} \}$.
Then it can be verified that $H$ is probabilistic observable by the definition.
   We could construct a probabilistic P-supervisor $S_{p}$ such that $L_{S_{p}/G} = L_{H}$ as follows.
  First of all, by Equation (36), the scaling-factor function $K$ could be computed:
  for $t_{1} \in (\sigma_{2} | \sigma_{1}  )^{*}\sigma_{3}$,
         $$K(t_{1}) = ( 0.8 \quad 1 \quad 1 \quad 1 \quad 1)^{\mathsf {T}}, $$
  and for $\forall t_{2} \in P(supp(L(G)))\backslash (\sigma_{2} | \sigma_{1}  )^{*}\sigma_{3} $,
         $$K(t_{2}) = ( 1 \quad 1 \quad 1 \quad 1 \quad 1 )^{\mathsf {T}}. $$
  Secondly, by solving the Equation (21), we could obtain one of the probabilistic P-supervisors as follows:
        $$S_{p}(t_{1}) = (0 \quad 0 \quad 0.2 \quad 0.8)^{\mathsf {T}},$$
        $$S_{p}(t_{2}) =(0 \quad 0 \quad 0 \quad 1 )^{\mathsf {T}}.$$
\end{example}

\subsection{Verification Algorithms of the probabilistic controllability and observability }
In Example 3, we have illustrated how to verify the probabilistic controllability and observability by definitions.
However, it is difficult to do so in a large scale system.
Hence, we would present two polynomial algorithms to verify the probabilistic controllability and observability in this subsection.
 \par
 First, we present a verification algorithm for the probabilistic controllability as follows.
 \par
\begin{algorithm}
 Given a plant $G = \{ X, x_{0}, \Sigma, \delta, \rho  \}$ with the controllable events set $\Sigma_{c}$,
    and the probabilistic specification $L$, such that $L \subseteq  L_{G}$.
    Suppose $L$ is generated by the probabilistic automaton $H = \{ Q, q_{0}, \Sigma, \delta_{H}, \rho_{H}  \}$, that is, $L = L_{H}$.
    \begin{enumerate}
      \item
      Construct the testing automaton $G_{tc}$ for the probabilistic controllability  as follows.
        \begin{equation}
            G_{tc} = \{ (X \times Q) \cup \{ d \}, (x_0, q_0), \Sigma, \delta_{tc} \}.
        \end{equation}
        $(X \times Q) \cup \{ d \}$ is the set of states.
    The (partial) transition function $\delta_{tc} : (X \times Q)  \times \Sigma \rightarrow (X \times Q) \cup \{ d \} $ is defined as follows.
     \begin{align}
          \delta_{tc}((x,q),\sigma) =
              \begin{cases}
                       (\delta(x,\sigma), \delta_{H}(q,\sigma)),           & \text{ if }    c_{1},     \\
                       d,                                                  & \text{ if }    c_{2}.
              \end{cases}
      \end{align}
     \par
     Here the $c_{1}$ denotes the condition:
        $[\delta_{H}(q,\sigma)! \wedge  \delta(x,\sigma)! \wedge \sigma \in \Sigma_{uc} \wedge \rho(x,\sigma) = \rho_{H}(q,\sigma)]
            \vee
        [\delta_{H}(q,\sigma)! \wedge  \delta(x,\sigma)! \wedge  \sigma \in \Sigma_{c} ]$; 
     and $c_{2}$ denotes the condition:
      $[  \sigma \in \Sigma_{uc} \wedge \rho(x,\sigma) \neq \rho_{H}(q,\sigma)]$

      \item
      Check whether or not the state ``$d$" is reachable from the initial state $(x_0,q_0)$.
      If the answer is yes, then $L_{H}$ is not probabilistic controllable;
      otherwise, $L_{H}$ is probabilistic controllable.
    \end{enumerate}

\end{algorithm}

\par
The basic idea of Algorithm 1 is capturing all the violations of the probabilistic controllability by reaching the state ``$d$" of $G_{tc}$.
Note that $|G_{tc}| = ((|X|*|Q|+1)*|\Sigma|)$. As a result, the complexity of Algorithm 1 is $O(|X|*|Q|*|\Sigma|)$.

\par
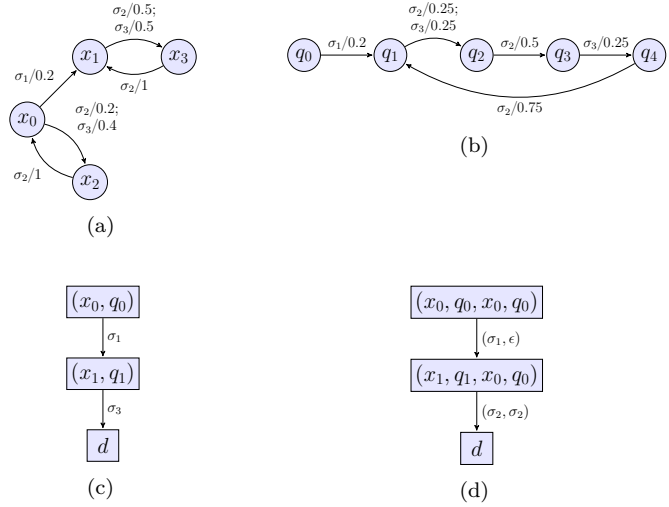
\begin{figure}[htbp]
\centering
\begin{tabular}{C{.25\textwidth}C{.5\textwidth}}

\subfigure [ ] {
    \resizebox{0.21\textwidth}{!}{%
                \centering
              \begin{tikzpicture}[scale=0.25,->,>=stealth',shorten >=1pt,auto,node distance=2.5cm, semithick,
                              every state/.style={fill=blue!10,thick,font=\LARGE}]

                            \node[state] (x0) {$x_0$};
                            \node[state] [above right  of = x0] (x1) {$x_1$};
                            \node[state] [below right of = x0] (x2) {$x_2$};
                            \node[state] [right       of = x1] (x3) {$x_3$};

                            \draw[every node/.style={font=\large}][align=center]
                                   (x0) edge  node{$\sigma_{1}/0.2$} (x1)
                                   (x0) edge [bend left] node {$\sigma_{2}/0.2;$ \\ $ \sigma_{3}/0.4$} (x2)
                                   (x2) edge [bend left] node {$\sigma_{2}/1$} (x0)
                                   (x1) edge [bend left] node{$\sigma_{2}/0.5;$ \\ $ \sigma_{3}/0.5$} (x3)
                                   (x3) edge [bend left] node{$\sigma_{2}/1$} (x1);
                            \end{tikzpicture}
    }
} &
\subfigure [ ] {
 \resizebox{0.42\textwidth}{!}{%
           \centering
              \begin{tikzpicture}[scale=0.25,->,>=stealth',shorten >=1pt,auto,node distance=2.5cm, semithick,
                              every state/.style={fill=blue!10,thick,font=\LARGE}]

                            \node[state] (q0) {$q_0$};
                            \node[state] [right  of = q0] (q1) {$q_1$};
                            \node[state] [right  of = q1] (q2) {$q_2$};
                            \node[state] [right  of = q2] (q3) {$q_3$};
                            \node[state] [right  of = q3] (q4) {$q_4$};

                            \draw[every node/.style={font=\large}][align=center]
                                   (q0) edge  node {$\sigma_{1}/0.2$} (q1)
                                   (q1) edge  [bend left] node {$\sigma_{2}/0.25;$ \\ $\sigma_{3}/0.25$} (q2)
                                   (q2) edge  node {$\sigma_{2}/0.5$} (q3)
                                   (q3) edge  node {$\sigma_{3}/0.25$} (q4)
                                   (q4) edge [bend left] node{$\sigma_{2}/0.75$} (q1);

                            \end{tikzpicture}
    }
 }\\
    \subfigure [ ] {
         \resizebox{0.085\textwidth}{!}{%
                    \centering
                           \begin{tikzpicture}[scale=0.25,->,>=stealth',shorten >=1pt,auto,node distance=2cm, semithick,
                              every state/.style={fill=blue!10,thick,rectangle,font=\LARGE}]

                            \node[state] (x0) {$(x_0, q_0)$};
                            \node[state] [below   of = x0] (x1) {$(x_1, q_1)$};
                            \node[state] [below   of = x1] (d) {$d$};

                            \draw[every node/.style={font=\large}]
                                   (x0) edge  node{$\sigma_{1}$} (x1)
                                   (x1) edge  node{$\sigma_{3}$} (d) ;
                            \end{tikzpicture}
         }
    }\qquad\qquad\qquad\qquad\qquad &
    \subfigure [ ] {
         \resizebox{0.15\textwidth}{!}{%
                    \centering
                           \begin{tikzpicture}[scale=0.25,->,>=stealth',shorten >=1pt,auto,node distance=2cm, semithick,
                              every state/.style={fill=blue!10,thick,rectangle,font=\LARGE}]

                            \node[state] (x0) {$(x_0, q_0,x_0, q_0)$};
                            \node[state] [below   of = x0] (x1) {$(x_1, q_1,x_0, q_0)$};
                            \node[state] [below   of = x1] (d) {$d$};

                            \draw[every node/.style={font=\large}]
                                   (x0) edge  node{$(\sigma_{1}, \epsilon)$} (x1)
                                   (x1) edge  node{$(\sigma_{2}, \sigma_{2})$ } (d) ;
                            \end{tikzpicture}
         }
    }
\end{tabular}
\caption{(a). $G$ in Example 4; (b). $H$ in Example 4; (c). A part of $G_{tc}$ in Example 4; (d). A part of $G_{to}$ in Example 5.}
\end{figure}

\par

\begin{example}
  Given a plant $G = \{ X, x_{0}, \Sigma, \delta, \rho  \}$, $\Sigma = \{ \sigma_{1}, \sigma_{2}, \sigma_{3}\}$, with
  $\Sigma_{c} = \{\sigma_{1}, \sigma_{2}\}$ and $\Sigma_{o} = \{\sigma_{2}, \sigma_{3}\}$,
    and the specification $H = \{ Q, q_{0}, \Sigma, \delta_{H}, \rho_{H}  \}$, as shown in Fig. 3 - (a) and (b), respectively.
According to Algorithm 1, we could construct the testing automaton $G_{tc}$. A part of $G_{tc}$ is shown in Fig. 3 - (c).
Note that the state ``$d$" is reachable from $(x_{0}, q_{0})$.
Hence, the specification $L_{H}$ is not probabilistic controllable.
The reachability of the state ``$d$" in Fig. 3 - (c) explains  the fact that
there exists an $s = \sigma_{1}$ such that $\delta(x_{0},s) = x_{1}$ and $\delta_{H}(q_{0},s) = q_{1}$, and for $\sigma_{3} \in \Sigma_{uc}$, $0.5 = \rho(x_{1},\sigma_{3}) \neq \rho_{H}(q_{1},\sigma_{3}) = 0.25$, which violates the definition of the probabilistic controllability.

\end{example}

\par
The following theorem demonstrates the correctness of Algorithm 1.

\par

 \begin{theorem}
    Given a plant $G = \{ X, x_{0}, \Sigma, \delta, \rho  \}$ with the controllable events set $\Sigma_{c}$,
        and the probabilistic specification $L$, such that $L \subseteq  L_{G}$.
    Suppose $L$ is generated by the probabilistic automaton $H = \{ Q, q_{0}, \Sigma, \delta_{H}, \rho_{H}  \}$, that is, $L = L_{H}$.
    $L$ is not probabilistic controllable if and only if the state ``$d$" of the automaton $G_{tc}$ is reachable from the initial state $(x_0,q_0)$.
 \end{theorem}
 \begin{proof}
    According to the definition of $G_{tc}$, if the state ``$d$" is reachable from the initial state $(x_0,q_0)$,
    then there exists an $s \in supp(L_{H})$, such that $\delta_{H}(q_{0},s) = q$ and $\delta(x_{0},s) = x$, and $\sigma \in \Sigma$,  the following condition holds:
      $$ \sigma \in \Sigma_{uc} \wedge \rho(x,\sigma) \neq \rho_{H}(q,\sigma).$$

    \par
    Therefore, by the definition of the probabilistic controllability, $L_{H}$ is not probabilistic controllable.

    \par
    Conversely, if $L_{H}$ is not probabilistic controllable, then by the definition of the probabilistic controllability,
    there exist $s \in supp(L_{H})$ and $\sigma \in \Sigma_{uc}$,
     such that $\delta( x_0 , s ) = x $ and $\delta_{H}( q_0 , s ) = q $,
    and            $\rho_{H}(q,\sigma) \neq \rho(x,\sigma)$.
    Hence, by the definition of $G_{tc}$, the state ``$d$" is reachable from $(x,q)$.
    If $(x,q)$ can be reachable from $(x_0,q_0)$, then ``$d$" is reachable from $(x_0,q_0)$.
    If $(x,q)$ cannot be reachable from $(x_0,q_0)$, then by the definition of $G_{tc}$, there must exist an $s^{'} \in \overline{s}$ such that
    state ``$d$"  is reachable by $s^{'}$ from $(x_0,q_0)$.
 \end{proof}

\par
In the following, we present a verification algorithm for the probabilistic observability, and then prove
its correctness.

\par

\begin{algorithm}
 Given a plant $G = \{ X, x_{0}, \Sigma, \delta, \rho  \}$ with the controllable events set $\Sigma_{c}$ and the observable events set $\Sigma_{o}$,
    and the specification $H = \{ Q, q_{0}, \Sigma, \delta_{H}, \rho_{H}  \}$, such that $L_{H} \subseteq  L_{G}$.
    \begin{enumerate}
      \item
       Construct the testing automaton $G_{to}$ for the probabilistic observability  as follows.
     \begin{multline}
      G_{to} = \{ (X \times Q \times X \times Q ) \cup \{ d \}, (x_0, q_0,x_0, q_0), \\
                                (\Sigma \cup \epsilon) \times (\Sigma \cup \epsilon), \delta_{to} \}.
        \end{multline}
    Here $(X \times Q \times X \times Q ) \cup \{ d \}$ is the set of states.
     The (partial) transition function $\delta_{to} : (X \times Q \times X \times Q ) \times ((\Sigma \cup \epsilon) \times (\Sigma \cup \epsilon)) \rightarrow (X \times Q \times X \times Q ) \cup \{ d \} $ is defined as follows.
    \begin{itemize}
            \item
                For each $\sigma \in \Sigma$,
                 \begin{align}
                   \delta_{to}((x_{1},q_{1},x_{2},q_{2}),(\sigma,\sigma)) = \qquad\qquad\qquad\qquad\quad & \nonumber \\
                          \begin{cases}
                                   (\delta(x_1,\sigma),\delta_{H}(q_1,\sigma),\delta(x_2,\sigma),\delta_{H}(q_2,\sigma)), &\text{if } c_{1}^{'}, \\
                                   d,                                                  & \text{if }   c_{2}^{'}.
                          \end{cases}
                  \end{align}
             Here the $c_{1}^{'}$ denotes  the condition: $[\delta(x_1,\sigma)!]$ $\wedge$ $[\delta_{H}(q_1,\sigma)!]$ $\wedge$ $[\delta(x_2,\sigma)!]$ $\wedge$ $[\delta_{H}(q_2,\sigma)!]$
             $\wedge$
             $[ ( \rho(x_{1},\sigma) * \rho_{H}(q_{2},\sigma)  = \rho(x_{2},\sigma) * \rho_{H}(q_{1},\sigma) \wedge \sigma \in \Sigma_{c} ) \vee
              \sigma \in \Sigma_{uc} ]$;
             $c_{2}^{'}$ denotes  the condition: $[\sigma \in \Sigma_{c}]$
             $\wedge$ $[\rho(x_{1},\sigma) * \rho_{H}(q_{2},\sigma)  \neq \rho(x_{2},\sigma) * \rho_{H}(q_{1},\sigma)]$.

            \item
                Particularly, if $\sigma \in \Sigma_{uo}$, the additional transitions are defined as follows.
                    \begin{subequations}
                          \begin{gather}
                                      \begin{align}
                                                \delta_{to}((x_{1},q_{1},x_{2},q_{2}),(\sigma,\epsilon)) =
                                                (\delta(x_1,\sigma), \delta_{H}(q_1,\sigma),x_{2}, q_{2}), & \text{ if }    c_{3}^{'},\\
                                                \delta_{to}((x_{1},q_{1},x_{2},q_{2}),(\epsilon,\sigma)) =
                                                (x_1, q_1,\delta(x_2,\sigma), \delta_{H}(q_2,\sigma)),     & \text{ if }    c_{4}^{'}.
                                      \end{align}
                          \end{gather}
                    \end{subequations}
                Here the $c_{3}^{'}$ denotes  the condition: $[\delta(x_1,\sigma)!]$ $\wedge$ $[\delta_{H}(q_1,\sigma)!]$;
                     the $c_{4}^{'}$ denotes the condition:  $[\delta(x_2,\sigma)!]$ $\wedge$ $[\delta_{H}(q_2,\sigma)!]$.
    \end{itemize}

      \item
       Check whether or not the state ``$d$" is reachable from the initial state $(x_0,q_0,x_0,q_0)$.
      If the answer is yes, then $L_{H}$ is not probabilistic observable;
      otherwise, $L_{H}$ is probabilistic observable.

    \end{enumerate}

\end{algorithm}

\par
 Similar to Algorithm 1, the basic idea of Algorithm 2 is
capturing all the violations of the probabilistic observability by reaching the state ``$d$" in automaton $G_{to}$.
Note that $|G_{to}| = ((|X|^{2}*|Q|^{2}+1)*(|\Sigma|+1)^{2}$.
Hence, the complexity of Algorithm 2 is $O(|X|^{2}*|Q|^{2}*|\Sigma|^{2})$.

\par
The following example illustrates how to verify the  probabilistic observability according to Algorithm 2.

\par
\begin{example}
  The plant $G = \{ X, x_{0}, \Sigma, \delta, \rho  \}$, and specification
  $H = \{ Q, q_{0}, \Sigma, \delta_{H}, \rho_{H}  \}$ considered here are the same as those adopted in Example 4.
  \par

  According to Algorithm 2, we could construct the testing automaton $G_{to}$ for probabilistic observable.
  A part of $G_{to}$ is shown in Fig. 3 - (d).
  Note that the state ``$d$" is reachable from $(x_0, q_0, x_0, q_0)$.
  Hence, the specification is not probabilistic observable.
  The reachability of the state ``$d$" in Fig. 3 - (d) explains the fact that
   there exist $s_{1} = \sigma_{1}$ and $s_{2} = \epsilon$, $P(s_1) = P(s_2)$, such that
   $\delta(x_{0},s_1) = x_{1}$,  $\delta_{H}(q_{0},s_1) = q_{1}$,
    and $\delta(x_{0},s_2) = x_{0}$, $\delta_{H}(q_{0},s_2) = q_{0}$,
    and for $\sigma_{2} \in \Sigma_{c}$, $0 = \rho(x_{1},\sigma_{2}) * \rho_{H}(q_{0},\sigma_{2}) \neq \rho(x_{0},\sigma_{2}) * \rho_{H}(q_{1},\sigma_{2}) = 0.05$, which violates the definition of the probabilistic observability.

\end{example}

\par
The following theorem  demonstrates the correctness of Algorithm 2.

\par

 \begin{theorem}
    Given a plant $G = \{ X, x_{0}, \Sigma, \delta, \rho  \}$ with the controllable events set $\Sigma_{c}$ and the observable events set $\Sigma_{o}$,
    and the specification $H = \{ Q, q_{0}, \Sigma, \delta_{H}, \rho_{H}  \}$.
    $L_{H}$ is not probabilistic observable if and only if the state ``$d$" of the automaton $G_{to}$ is reachable from the initial state $(x_0,q_0,x_0,q_0)$.
 \end{theorem}

 \begin{proof}
    Suppose $(s_1, s_2) $ is the string tuple that reaches $(x_1, q_1, x_2, q_2)$ from the initial state of $G_{to}$: $(x_0, q_0,x_0,q_0)$.
    Assume $\delta(x_0,s_i) = x_i$ and $\delta_{H}(q_0,s_i) = q_i$, $i \in \{1,2\}$.
    According to the definition of $G_{to}$, we have $s_1, s_2 \in supp(L_{H})$, and $P(s_{1}) = P(s_2)$.
    If state ``$d$" is reachable by event $\sigma$ from the reachable state $(x_1, q_1, x_2, q_2)$,
    then we obtain $\sigma \in \Sigma_{c}$ and
    $\rho(x_{1},\sigma) * \rho_{H}(q_{2},\sigma)  \neq \rho(x_{2},\sigma) * \rho_{H}(q_{1},\sigma)$.
    Hence, $L_{H}$ is not probabilistic observable.

    \par
    On the other hand, if $L_{H}$ is not probabilistic observable, then according to the definition of the probabilistic observability,
    there exist $s_{1},s_{2} \in supp(L_{H})$, such that $P(s_{1}) = P(s_{2})$, and
    $\delta( x_0 , s_{i} ) = x_{i} $ and $\delta_{H}( q_0 , s_{i} ) = q_{i} $, $i = \{1 , 2\}$,
    and $\exists \sigma \in \Sigma_{c}$,
   $ \rho(x_{1},\sigma) * \rho_{H}(q_{2},\sigma)  \neq \rho(x_{2},\sigma) * \rho_{H}(q_{1},\sigma)$.
   By the definition of $G_{to}$, the state ``$d$" is reachable from $(x_1,q_1,x_2,q_2)$ by $\sigma$.
   If $(x_1,q_1,x_2,q_2)$ can be reachable from $(x_0,q_0,x_0,q_0)$,
   then the state ``$d$" can be reachable from $(x_0,q_0,x_0,q_0)$;
   otherwise, by the definition of $G_{to}$,
   there must exist $s_1^{'} \in \overline{s_1}$, $s_2^{'} \in \overline{s_2}$ such that $P(s_1^{'}) = P(s_2^{'})$,
   the state ``$d$" is reachable from $(x_0, q_0, x_0, q_0)$ by $(s_1^{'}, s_2^{'})$.
 \end{proof}

\section{Infimal probabilistic controllable and observable superlanguage}
When the given specification is unachievable (not probabilistic controllable or observable), it is natural to pursue the ``best" achievable approximation.
In this section, we present an general procedure to compute the \emph{infimal probabilistic controllable and observable superlanguage} for an unachievable sublanguage.

\par
In the last section, we have presented the definitions of probabilistic controllability and observability in the context of automata form.
For the convenience of the discussion in this section, the equivalent notions in the context of languages form are defined as follows.
\par

\begin{definition}
Let $L$ and $M$ be probabilistic languages over events set $\Sigma$, and $L \subseteq M$.
Let $\Sigma_{c}$ be the controllable events set.
   $L$ is said to be probabilistic controllable w.r.t. $M$ and $\Sigma_{c}$,
   if $\forall s \in supp(L)$, $\forall \sigma \in \Sigma_{uc}$,
        \begin{equation}
           \frac{L(s\sigma)}{L(s)} = \frac{M(s\sigma)}{M(s)}.
    \end{equation}
\end{definition}

\begin{definition}
   Let $L$ and $M$ be probabilistic languages over events set $\Sigma$, and $L \subseteq M$.
Let $\Sigma_{c}$ be the controllable events set, and $\Sigma_{o}$ be the observable events set.
   $L$ is said to be probabilistic observable w.r.t. $M$, $\Sigma_{c}$ and $\Sigma_{o}$,
     if $\forall s_{1},s_{2} \in supp(L)$, $P(s_{1}) = P(s_{2})$,
    and $\forall \sigma \in \Sigma_{c}$,
    \begin{equation}
            \frac{M(s_{1}\sigma)}{M(s_{1})} * \frac{L(s_{2}\sigma)}{L(s_{2})} =
            \frac{M(s_{2}\sigma)}{M(s_{2})} * \frac{L(s_{1}\sigma)}{L(s_{1})}.
    \end{equation}
\end{definition}

\par
According to the definitions, the following propositions could be obtained immediately.

\begin{proposition}
  Suppose $M$ and $L$ are generated by probabilistic automata $G$ and $H$, respectively.
  $L$ is probabilistic controllable (observable) w.r.t. $M$, $\Sigma_{c}$ (and $\Sigma_{o}$) $\Leftrightarrow$
  $H$ is probabilistic controllable (observable) w.r.t. $G$, $\Sigma_{c}$ (and $\Sigma_{o}$).
\end{proposition}

\begin{proposition}
  $L$ is probabilistic controllable (observable) w.r.t. $M$, $\Sigma_{c}$ (and $\Sigma_{o}$) $\Rightarrow$
  $supp(L)$ is controllable (observable) w.r.t. $supp(M)$, $\Sigma_{c}$ (and $\Sigma_{o}$).
\end{proposition}

\par
We present the definition of the class of \emph{probabilistic controllable and observable superlanguages} as follows.
\begin{multline}
  PCO(L) = \{ R | L \subseteq R \subseteq M \wedge R \text{ is probabilistic} \\
       \text{controllable and observable w.r.t. } M, \Sigma_{c} \text{ and } \Sigma_{o}  \}.
\end{multline}
\par
It is obvious that this class is not empty, as $M \in PCO(L)$.
In general, people are more interested in the infimal element of the class $PCO(L)$, as it could be viewed the ``best" achievable approximation of an unachievable sublanguage.
\par
Formally, the \emph{infimal probabilistic controllable and observable superlanguage} of $L$, denoted as $\inf \{ PCO(L)\}$, could be defined as follows.
\begin{multline}
        [ \inf \{ PCO(L) \} \in PCO(L) ] \wedge
          [ \forall R \in PCO(L), \inf \{ PCO(L) \} \subseteq R].
\end{multline}

\par
Firstly, it is necessary to investigate the existence of the infimal element $\inf \{ PCO(L) \}$.

\par%

\par
The following two propositions show that
the probabilistic controllability and observability are both closed under the  intersection operation (``$\cap$") of probabilistic languages.

\begin{proposition}
  Given three probabilistic languages $L_{1}$, $L_{2}$ and $M$ over events set $\Sigma$.
  If $L_{1}$ and $L_{2}$ are both probabilistic controllable w.r.t. $M$  and $\Sigma_{c}$,
  then so is $L_{1} \cap L_{2}$.
\end{proposition}
\begin{proof}
 By the definition of $L_{1} \cap L_{2}$,
 we have, $\forall s \in supp(L_{1} \cap L_{2})$,
    \begin{equation}
            \frac{(L_{1} \cap L_{2})(s\sigma)}{(L_{1} \cap L_{2})(s)}    =  \min \{ \frac{L_{1}(s\sigma)}{L_{1}(s)}, \frac{L_{2}(s\sigma)}{L_{2}(s)} \}.
    \end{equation}

    \par
  By means of the probabilistic controllabilities of $L_{1}$ and $L_{2}$, we obtain, $\forall \sigma \in \Sigma_{uc}$,
  \[
    \frac{(L_{1} \cap L_{2})(s\sigma)}{(L_{1} \cap L_{2})(s)}    = \min \{ \frac{M(s\sigma)}{M(s)}, \frac{M(s\sigma)}{M(s)} \} = \frac{M(s\sigma)}{M(s)}.
  \]
    Hence, $L_{1} \cap L_{2}$ is  probabilistic controllable.
\end{proof}

\par

\begin{proposition}
  Given three probabilistic languages $L_{1}$, $L_{2}$ and $M$ over events set $\Sigma$.
  If $L_{1}$ and $L_{2}$ are both probabilistic observable w.r.t. $M$, $\Sigma_{c}$  and $\Sigma_{o}$,
  then so is $L_{1} \cap L_{2}$.
\end{proposition}

\begin{proof}
    For $ s_{1}, s_{2}  \in supp(L_{1} \cap L_{2})$, $P(s_{1}) = P(s_{2})$, and $\sigma \in \Sigma_{c}$, we have
     \begin{align}
     & \frac{M(s_{1}\sigma)}{M(s_{1})} * \frac{(L_{1} \cap L_{2})(s_{2}\sigma)}{(L_{1} \cap L_{2})(s_{2})}           \nonumber \\
         &=    \frac{M(s_{1}\sigma)}{M(s_{1})} * \min \{ \frac{L_{1}(s_{2}\sigma)}{L_{1}(s_{2})} , \frac{L_{2}(s_{2}\sigma)}{L_{2}(s_{2})} \} \text{ (By the definition}   \text{ of $\cap$) }\nonumber \\
         &=  \min \{  \frac{M(s_{1}\sigma)}{M(s_{1})} * \frac{L_{1}(s_{2}\sigma)}{L_{1}(s_{2})}, \frac{M(s_{1}\sigma)}{M(s_{1})} * \frac{L_{2}(s_{2}\sigma)}{L_{2}(s_{2})} \} \nonumber \\
         &=  \min \{  \frac{M(s_{2}\sigma)}{M(s_{2})} * \frac{L_{1}(s_{1}\sigma)}{L_{1}(s_{1})}, \frac{M(s_{2}\sigma)}{M(s_{2})} * \frac{L_{2}(s_{1}\sigma)}{L_{2}(s_{1})} \} \text{ (By} \nonumber \\
         &    \text{ the probabilistic  observabilities of $L_{1}$ and $L_{2}$) }\nonumber \\
         &= \frac{M(s_{2}\sigma)}{M(s_{2})} * \min \{  \frac{L_{1}(s_{1}\sigma)}{L_{1}(s_{2})}, \frac{L_{2}(s_{1}\sigma)}{L_{2}(s_{1})} \} \nonumber \\
         &= \frac{M(s_{2}\sigma)}{M(s_{2})} * \frac{(L_{1} \cap L_{2})(s_{1}\sigma)}{(L_{1} \cap L_{2})(s_{1})} \text{ (By the definition of $\cap$) }
            \end{align}
Therefore, $L_{1} \cap L_{2}$ is probabilistic observable.
\end{proof}

\par
Since the class $PCO(L)$ is not empty, as mentioned before,
Propositions 1, 5 and 6 can guarantee the existence of $\inf \{ PCO(L) \}$.

 \par
In the rest of this section, we would focus on the computation of $\inf \{ PCO(L) \}$.
Since the support language of a probabilistic language is always prefix-closed,
the non-probabilistic version of $PCO(L)$ could be defined as follows.
\begin{multline}
  CO(supp(L)) = \{ R_{s} | supp(L) \subseteq R_{s} \subseteq supp(M) \wedge R_{s} = \overline{R_{s}} \wedge \\
R_{s} \text{ is  controllable and observable }
       \text{w.r.t. } supp(M), \sigma_{c} \text{ and } \Sigma_{o}\}.
\end{multline}
\par

By means of the definition of sublanguage and Proposition 4, we have the following proposition immediately.
\begin{proposition}
The support language of the element in $PCO(L)$ must be in $CO(supp(L))$, that is,
    \begin{equation*}
      R \in PCO(L) \Rightarrow supp(R) \in  CO(supp(L)).
    \end{equation*}
\end{proposition}

\par
Actually, the set $CO(\cdot)$ defined above is called as the \emph{prefix-closed controllable and observable superlanguages} in literature.
It has been well investigated in \cite{Rudie1}, \cite{Shayman} and \cite{Masopust}.
Hence, we could compute $\inf \{ PCO(L) \}$ by means of some known results in the non-probabilistic situation.
\par

\par
By means of the approach introduced by Masopust \cite{Masopust},
we can obtain the finite automaton $H_{s}$ that generates the infimal element of $CO(supp(L))$,
that is, $L_{H_{s}} = \inf \{ CO(supp(L)) \}$.

\par
Suppose the probabilistic languages $M$ and $L$ are generated by the probabilistic automata $G$ and $H$, respectively. That is, $L_{G} = M$ and $L_{H} = L$.
\par
In order to simplify the computation of  $\inf \{ PCO(L) \}$,
it is necessary to refine the automaton $H_{s}$, and the probabilistic automata $G$ and $H$.
Our final goal is to obtain two \emph{normal probabilistic automata} $G_{n}$ and $H_{n}$, $H_{n} \sqsubseteq  G_{n}$,
such that $L_{H_{n}}(s) = L(s)$ for $s \in supp(L)$ and $supp(L_{H_{n}}) = \inf \{ CO(supp(L))$, and $L_{G_{n}} = M$.

\par
We have the following explanations for the aforementioned refinements.
\begin{enumerate}

  \item
   The notion of \emph{normal automaton} and the normalization procedure were first proposed by Cho and Marcus \cite{Cho}.
   Normal probabilistic automata mentioned here are probabilistic automata whose logic parts are normal automata.

    \item
    Takai and Ushio \cite{Takai} pointed out an excellent property of normal automata:
     \emph{the state space of the observer  of a normal automaton is exactly a partition of the state space of the normal automaton}.

  \item
  $L_{H_{n}}(s) = L(s)$ for $s \in supp(L)$, and $supp(L_{H_{n}}) = \inf \{ CO(supp(L)) $ mean that
  $H_{n}$ encodes not only the quantitative information from $L$, but also the logic information from $\inf \{ CO(supp(L)) \}$.

  \item
  The refinement operation for the subautomaton relation ($H_{n} \sqsubseteq G_{n}$) can be realized efficiently.
  It helps us to simplify the subsequent computation for the automaton representation of $\inf \{ PCO(L) \}$.
\end{enumerate}

The refinement algorithm that output $G_{n}$ and $H_{n}$ is presented in Appendix A.
\par

\par
In the following, we would present an algorithm to compute an automaton representation for $\inf \{ PCO(L) \}$ based on
the normal probabilistic automata $H_{n}$ and $G_{n}$.

\par

\begin{algorithm}
  Given normal probabilistic automata $H_{n} = \{ X_{H_{n}}, x_{0,H_{n}},  \Sigma,\delta_{H_{n}}, \rho_{H_{n}} \} $ and $G_{n} = \{ X_{G_{n}}, x_{0,G_{n}},\Sigma,\delta_{G_{n}},  \rho_{G_{n}} \}$, $H_{n} \sqsubseteq  G_{n}$, such that
  $L_{H_{n}}(s) = L(s)$ for $s \in supp(L)$ and $supp(L_{H_{n}}) = \inf \{ CO(supp(L)) \}$,
   and $L_{G_{n}} = M$.

  \begin{enumerate}
  \item
   $\widetilde{H} \Leftarrow H_{n} $, and suppose $\widetilde{H} = \{ X_{\widetilde{H}}, x_{0,\widetilde{H}},  \Sigma,\delta_{\widetilde{H}}, \rho_{\widetilde{H}} \} $.

    \item
    Handle the transitions driven by uncontrollable events of $\widetilde{H}$.
    Specifically, for $\forall x \in X_{\widetilde{H}}$ and $\forall \sigma \in \Sigma_{uc}$, do the following.
                    \begin{enumerate}

                         \item
                       $\rho_{\widetilde{H}}(x,\sigma) \Leftarrow \rho_{G_{n}}(x,\sigma)$;

                         \item
                       If $\rho_{\widetilde{H}}(x,\sigma) > 0$ and $\rho_{H_{n}}(x,\sigma) = 0$ and $\sigma \in \Sigma_{uc}$,
                      then $ X_{\widetilde{H}} \Leftarrow  X_{\widetilde{H}} \cup \delta_{G_{n}}(x,\sigma)$ and
                       $\delta_{\widetilde{H}}(x,\sigma) \Leftarrow \delta_{G_{n}}(x,\sigma)$.

                    \end{enumerate}

    \item
     Handle the  transitions driven by controllable events of $\widetilde{H}$.
     Specifically, compute the observer $Obs(logic(H_{n}))$
     with the states set $X_{obs} = \{ X_{1}, X_{2}, \ldots, X_{r} \}$.
     Since $logic(H_{n})$ is normal, $X_{k_{1}} \cap X_{k_{2}} = \varnothing$ for $k_{1} \neq k_{2}$, $k_{1},k_{2} \in [1,r]$,
     and $\bigcup_{k \in [1,r]} X_{k} = X_{\widetilde{H}}$.
    For each $X_{k} = \{ x_1, x_2, \ldots,  x_{q} \}$, do the following.
    \begin{enumerate}
          \item
          Suppose $\Sigma_{c} = \{\sigma_{1}, \ldots, \sigma_{m}\}$.
          Compute $q$ vectors with $m$  size:
            $K(j) = [
            k_{j}^{1}  \quad
            k_{j}^{2}
            \ldots
            k_{j}^{m}
            ]^{T}$, $j \in [1,q]$, where $k_{j}^{i} = 0$, if $\rho_{H_{n}}(x_{j},\sigma_{i}) = 0$;
            otherwise, $k_{j}^{i} = \frac{\rho_{H_{n}}(x_{j},\sigma_{i})}{\rho_{G_{n}}(x_{j},\sigma_{i})}$.

             \item
                Compute the vector $K = [ k_{1}  \quad  k_{2} \ldots k_{m} ]^{T} $,
                where the $k_{i} = \max_{j = 1}^{q} \{ k_{j}^{i} \}$,
                 $i \in [1,m]$.

             \item
                 $\rho_{\widetilde{H} }(x_{j},\sigma_{i}) \Leftarrow k_{i} * \rho_{G_{n}}(x_{j},\sigma_{i})$,
                 $\forall j \in [1,q]$ and $\forall i \in [1,m]$.

             \item
                 If $\rho_{\widetilde{H}}(x_{j},\sigma_{i}) > 0$ and $\rho_{H_{n}}(x_{j},\sigma_{i}) = 0$,
                 then $\delta_{\widetilde{H}}(x_{j},\sigma_{i}) \Leftarrow \delta_{G_{n}}(x_{j},\sigma_{i})$ and $X_{\widetilde{H}} \Leftarrow X_{\widetilde{H}} \cup \delta_{G_{n}}(x_{j},\sigma_{i})$,
                 $\forall j \in [1,q]$ and $\forall i \in [1,m]$.

        \end{enumerate}
      \end{enumerate}
\end{algorithm}

\par

What the step 2) of Algorithm 3 do is searching each uncontrollable transitions in $G_{n}$ and $H_{n}$.
Since $H_{n} \sqsubseteq G_{n}$, the complexity of step 2) is $O(|X_{G_{n}}|*|\Sigma_{uc}|)$.
The step 3) of Algorithm 3 is based on observer automaton that has an exponential states space generally.
However, the normal automaton $logic(H_{n})$ has the following excellent property:
the state space of the observer of $logic(H_{n})$ is exactly a partition of the state space of $logic(H_{n})$ \cite{Takai}.
Hence, the complexity of step 3) is $O(|X_{G_{n}}|*|\Sigma_{c}|)$.
Therefore, the complexity of Algorithm 3 is $O(|X_{G_{n}}|*|\Sigma|)$.
\par

\begin{example}
 The normal probabilistic automaton $G_{n}$ that generates the behavior of the plant is  shown in Fig. 5 - (a).
 The normal probabilistic automaton $H_{n}$ that encodes the (quantitative) information from $L$, and the (logic)
information from $\inf \{ CO(supp(L)) \}$ is shown in Fig. 5 - (b).
The probabilistic automaton $H$, as the generator of the original specification $L$, is shown in Fig. 5 - (c).
Let $\Sigma_{c} = \{\sigma_{1}, \sigma_{2}\}$ and $\Sigma_{o} = \{\sigma_{2}, \sigma_{3}\}$.
\par
  According to Algorithm 3, we obtain $\widetilde{H}$, as shown in Fig. 5 - (d).

\begin{figure}[htbp]
\centering
\begin{tabular}{C{.4\textwidth}C{.4\textwidth}}
\subfigure [ ] {
    \resizebox{0.3\textwidth}{!}{%
                \centering
               \begin{tikzpicture}[scale=0.25,->,>=stealth',shorten >=1pt,auto,node distance=2.5cm, semithick,
                                  every state/.style={fill=blue!10,thick,font=\LARGE}]

                                \node[state] (x0) {$x_0$};
                                \node[state] [below right   of = x0] (x5) {$x_5$};
                                \node[state] [below left of = x0] (x1) {$x_1$};
                                \node[state] [below  of = x1] (x2) {$x_2$};
                                \node[state] [below  of = x2] (x3) {$x_3$};
                                \node[state] [below  of = x3] (x4) {$x_4$};
                                \node[state] [below  of = x5] (x6) {$x_6$};
                                \node[state] [below  of = x6] (x7) {$x_7$};

                                \draw[every node/.style={font=\large }][align=center]
                                       (x0) edge [bend right] node {$\sigma_{1}/0.2$} (x1)
                                       (x0) edge [bend left] node {$\sigma_{2}/0.2  $ \\ $  \sigma_{3}/0.4$} (x5)
                                       (x1) edge  node {$\sigma_{2}/0.5 $ \\ $ \sigma_{3}/0.5$} (x2)
                                       (x2) edge  node {$\sigma_{2}/1$} (x3)
                                       (x3) edge  node {$\sigma_{2}/0.5  $ \\ $  \sigma_{3}/0.5$} (x4)
                                       (x4) edge [bend left] node{$\sigma_{2}/1$} (x1)
                                       (x5) edge  node {$\sigma_{2}/1$} (x6)
                                       (x6) edge  node {$\sigma_{2}/0.5  $ \\ $  \sigma_{3}/0.5$} (x7)
                                       (x7) edge  node {$\sigma_{2}/0.8$} (x0);
                                \end{tikzpicture}
    }
}&
   \subfigure [ ] {
    \resizebox{0.323\textwidth}{!}{%
                \centering
               \begin{tikzpicture}[scale=0.25,->,>=stealth',shorten >=1pt,auto,node distance=2.5cm, semithick,
                                  every state/.style={fill=blue!10,thick,font=\LARGE}]

                                \node[state] (x0) {$x_0$};
                                \node[state] [below right   of = x0] (x5) {$x_5$};
                                \node[state] [below left of = x0] (x1) {$x_1$};
                                \node[state] [below  of = x1] (x2) {$x_2$};
                                \node[state] [below  of = x2] (x3) {$x_3$};
                                \node[state] [below  of = x3] (x4) {$x_4$};
                                \node[state] [below  of = x5] (x6) {$x_6$};
                                \node[state] [below  of = x6] (x7) {$x_7$};

                                \draw[every node/.style={font=\large}][align=center]
                                       (x0) edge [bend right] node {$\sigma_{1}/0.2$} (x1)
                                       (x0) edge [bend left] node {$\sigma_{2}/0^{+}$ \\ $  \sigma_{3}/0^{+}$} (x5)
                                       (x1) edge  node {$\sigma_{2}/0.25  $ \\ $  \sigma_{3}/0.25$} (x2)
                                       (x2) edge  node {$\sigma_{2}/0.5$} (x3)
                                       (x3) edge  node {$\sigma_{3}/0.25$} (x4)
                                       (x4) edge [bend left] node{$\sigma_{2}/0.75$} (x1)
                                       (x5) edge  node {$\sigma_{2}/0^{+}$} (x6)
                                       (x6) edge  node {$\sigma_{2}/0^{+}$ \\ $ \sigma_{3}/0^{+}$} (x7)
                                       (x7) edge  node {$\sigma_{2}/0^{+}$} (x0);

                                \end{tikzpicture}
    }
} \\
\subfigure [ ] {
 \resizebox{0.21\textwidth}{!}{%
           \centering
              \begin{tikzpicture}[scale=0.25,->,>=stealth',shorten >=1pt,auto,node distance=2.5cm, semithick,
                              every state/.style={fill=blue!10,thick,font=\LARGE}]

                            \node[state] (x0) {$x_0$};
                            \node[state] [below left of = x0] (x1) {$x_1$};
                            \node[state] [below  of = x1] (x2) {$x_2$};
                            \node[state] [below  of = x2] (x3) {$x_3$};
                            \node[state] [below  of = x3] (x4) {$x_4$};

                            \draw[every node/.style={font=\large}][align=center]
                                   (x0) edge  [bend right] node {$\sigma_{1}/0.2$} (x1)
                                   (x1) edge  node {$\sigma_{2}/0.25$ \\ $  \sigma_{3}/0.25$} (x2)
                                   (x2) edge  node {$\sigma_{2}/0.5$} (x3)
                                   (x3) edge  node {$\sigma_{3}/0.25$} (x4)
                                   (x4) edge [bend left] node{$\sigma_{2}/0.75$} (x1);
                            \end{tikzpicture}
    }
 }&
 \subfigure [ ] {
         \resizebox{0.305\textwidth}{!}{%
                    \centering
                    \begin{tikzpicture}[scale=0.25,->,>=stealth',shorten >=1pt,auto,node distance=2.5cm, semithick,
                                  every state/.style={fill=blue!10,thick,font=\LARGE}]

                                \node[state] (x0) {$x_0$};
                                \node[state] [below right   of = x0] (x5) {$x_5$};
                                \node[state] [below left of = x0] (x1) {$x_1$};
                                \node[state] [below  of = x1] (x2) {$x_2$};
                                \node[state] [below  of = x2] (x3) {$x_3$};
                                \node[state] [below  of = x3] (x4) {$x_4$};
                                \node[state] [below  of = x5] (x6) {$x_6$};
                                \node[state] [below  of = x6] (x7) {$x_7$};

                                \draw[every node/.style={font=\large}][align=center]
                                       (x0) edge [bend right] node {$\sigma_{1}/0.2$} (x1)
                                       (x0) edge [bend left] node {$\sigma_{2}/0.1$ \\ $  \sigma_{3}/0.4$} (x5)
                                       (x1) edge  node {$\sigma_{2}/0.25$ \\ $  \sigma_{3}/0.5$} (x2)
                                       (x2) edge  node {$\sigma_{2}/0.5$} (x3)
                                       (x3) edge  node {$\sigma_{3}/0.5$} (x4)
                                       (x4) edge [bend left] node{$\sigma_{2}/0.75$} (x1)
                                       (x5) edge  node {$\sigma_{2}/0.5$} (x6)
                                       (x6) edge node {$\sigma_{3}/0.5$} (x7)
                                       (x7) edge node {$\sigma_{2}/0.6$} (x0);

                                \end{tikzpicture}
         }
    }
\end{tabular}
\caption{(a). $G_{n}$ (the normal generator for plant behavior); (b). $H_{n}$ (the normal automaton that encodes the information in specification $L$ and $\inf\{supp(L)\}$);
 (c). $H$ (the generator of $L$); (d). $\widetilde{H}$ (the output of Algorithm 3).}
\end{figure}
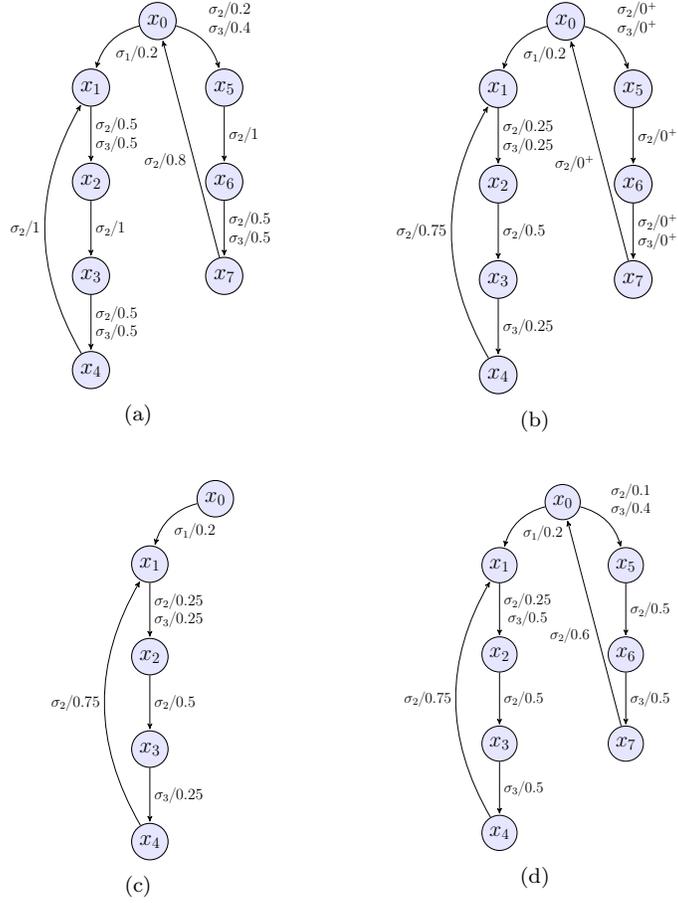

\end{example}

\par
At the end of this section, we would prove that $\widetilde{H}$ is the generator of $\inf \{ PCO(L)\}$,
that is, $L_{\widetilde{H}} = \inf \{ PCO(L)\} $.
\par
We present several necessary Lemmas as follows.
\begin{lemma}
  $logic( \widetilde{H} ) = logic( H_{n} ) $.
\end{lemma}
\begin{proof}
Obviously, we need to show that the steps 2) and 3) of Algorithm 3 do not bring any new transition and state to $\widetilde{H}$.
It is sufficient to prove that the conditions in step 2) - b) and step 3) - d) both always do not hold.
 \par

  We first prove the condition in step 2) - b), that is
   $\rho_{\widetilde{H}}(x,\sigma) > 0$ and $\rho_{H_{n}}(x,\sigma) = 0$ and $\sigma \in \Sigma_{uc}$,
  always do not hold.
  It is equivalent to prove $\delta_{H_{n}}(x,\sigma)$\sout{!} and $\delta_{G_{n}}(x,\sigma)!$ and $\sigma \in \Sigma_{uc}$
  do not hold for $\forall x \in X_{H_{n}}$.
  For contradiction, suppose there exists an $x^{'} \in X_{H_{n}}$,  such that
  $\delta_{H_{n}}(x^{'},\sigma)$\sout{!} and $\delta_{G_{n}}(x^{'},\sigma)!$ and $\sigma \in \Sigma_{uc}$.
  Assume $\delta_{H_{n}}(x_{0,H_{n}}, s^{'})= \delta_{G_{n}}(x_{0,G_{n}}, s^{'}) = x^{'}$.
  This means $s^{'} \in supp(L_{H_{n}})$ and $s^{'}\sigma \not\in supp(L_{H_{n}})$ and $s^{'}\sigma \in supp(L_{G_{n}}) = supp(M)$.
  Thus, $supp(L_{H_{n}})$ is not controllable w.r.t $supp(M)$ and $\Sigma_{c}$.
  This contradicts $supp(L_{H_{n}}) = \inf \{ CO(supp(L)) \}$.
  \par

    We then prove that the condition in step 3) - d),
  that is $\rho_{\widetilde{H}}(x_{j},\sigma_{i}) > 0$ and $\rho_{H_{n}}(x_{j},\sigma_{i}) = 0$, do not hold for
  $\forall j \in [1,q]$ and $\forall i \in [1,m]$.
  For contradiction, suppose  there exist $j_{1} \in [1,q]$ and $i_{1} \in [1,m]$, such that
  $\rho_{\widetilde{H}}(x_{j_{1}},\sigma_{i_{1}}) >  0$ and $\rho_{H_{n}}(x_{j_{1}},\sigma_{i_{1}}) = 0$.
  According to the step 3) of Algorithm 3, $\rho_{G_{n}}(x_{j_{1}},\sigma_{i_{1}}) > 0$, and
   there exists $ j_{2} \in [1,q]$,  such that
  $k_{i_{1}} = k_{j_{2}}^{i_{1}} = \frac{\rho_{H_{n}}(x_{j_{2}},\sigma_{i_{1}})}{\rho_{G_{n}}(x_{j_{2}},\sigma_{i_{1}})} > 0$,
  and there exists $k \in [1,r]$, such that $x_{j_{1}}, x_{j_{2}} \in X_{k} \in X_{obs}$.
  That is, $\rho_{H_{n}}(x_{j_{2}},\sigma_{i_{1}}) > 0$, and there exist $s_{1}, s_{2} \in supp(L_{H_{n}})$
  such that $P(s_{2}) = P(s_{1})$, and $\delta_{H^{n}}(x_{0,H^{n}}, s_{1} ) = x_{j_{1}}$, and $\delta_{H^{n}}(x_{0,H^{n}}, s_{2} ) = x_{j_{2}}$.
  In summary, $s_{2}\sigma_{i_{1}}, s_{1} \in supp(L_{H_{n}})$,
  and $s_{1}\sigma_{i_{1}} \not\in supp(L_{H_{n}})$,
  and $s_{1}\sigma_{i_{1}} \in supp(L_{G_{n}}) = supp(M)$.
  It means $supp(L_{H_{n}})$ is not observable w.r.t. $supp(M)$, $\Sigma_{c}$ and $\Sigma_{o}$.
  This also contradicts $supp(L_{H_{n}}) = \inf \{ CO(supp(L)) \}$.
\end{proof}

\begin{lemma}
  $L_{\widetilde{H}} \in PCO(L)$.
\end{lemma}
\begin{proof}
We first prove that $L_{\widetilde{H}}$ is probabilistic controllable.
Given $s \in supp(L_{\widetilde{H}})$, and suppose that $ x_{j} = \delta_{\widetilde{H}}( x_{0,\widetilde{H}} , s ) = \delta_{G_{n}}( x_{0,G_{n}} , s ) $.
Since we have shown in Lemma 2 that the steps 2) and 3) of Algorithm 3 do not bring any new transition and state to $\widetilde{H}$,
 we have
$\rho_{\widetilde{H}}(x_{j},\sigma_{i}) = \rho_{G_{n}}(x_{j},\sigma_{i})$
for $\forall x_{j} \in X_{\widetilde{H}}$ and $\forall \sigma_{i} \in \Sigma_{uc}$,
 by the step 2) - a) of Algorithm 3.
Thus, $\widetilde{H}$ is probabilistic controllable w.r.t. $G_{n}$, $\Sigma_{c}$.
By Proposition 3, $L_{\widetilde{H}}$ is probabilistic controllable w.r.t. $L_{G_{n}}$, $\Sigma_{c}$.

\par

We then prove that $L_{\widetilde{H}}$ is probabilistic observable.
We take $\forall s_{1},s_{2} \in supp(L_{\widetilde{H}})$, such that $P(s_{1}) = P(s_{2})$,
and suppose $\delta_{\widetilde{H}}( x_{0,\widetilde{H}} , s_{j} ) = \delta_{G_{n}}( x_{0,G_{n}} , s_{j} ) = x_{j} $,
    $j = \{1 , 2\}$.
Since we have shown in Lemma 2 that the steps 2) and 3) of Algorithm 3 do not bring any new  state to $H$,
there exists $X_{k} \in X_{obs}$, such that $x_{1},x_{2} \in X_{k}$.
By the step 3) of Algorithm 3,  we have
$\rho_{\widetilde{H} }(x_{j},\sigma_{i}) = k_{i} * \rho_{G_{n}}(x_{j},\sigma_{i})$, $j = \{1,2\}$, for $\forall \sigma_{i} \in \Sigma_{c}$.
This implies $ \rho_{\widetilde{H} }(x_{1},\sigma_{i}) * \rho_{G_{n}}(x_{2},\sigma_{i}) =
                     \rho_{\widetilde{H} }(x_{2},\sigma_{i}) * \rho_{G_{n}}(x_{1},\sigma_{i})$.
Thus, $\widetilde{H}$ is probabilistic observable w.r.t. $G_{n}$, $\Sigma_{c}$ and $\Sigma_{o}$.
By Proposition 3, $L_{\widetilde{H}}$ is probabilistic observable w.r.t. $L_{G_{n}}$, $\Sigma_{c}$, and $\Sigma_{o}$.

\par
By Algorithm 3, we have $ H_{n} \sqsubseteq \widetilde{H} \sqsubseteq G_{n}$.
According to Proposition 2,
 we obtain $ L_{H_{n}} \subseteq L_{\widetilde{H}} \subseteq L_{G_{n}}$.

 \par
  Therefore, $L_{\widetilde{H}} \in PCO(L)$.
\end{proof}

\begin{lemma}
   $ supp( L_{\widetilde{H}} ) = supp( \inf \{ PCO(L)\}) $
\end{lemma}
\begin{proof}
Lemma 3 implies $\inf \{ PCO(L)\} \subseteq L_{\widetilde{H}} $.
 Hence, $supp( \inf \{ PCO(L)\}) \subseteq supp( L_{\widetilde{H}} ) $.
\par

Obviously, $\inf \{ PCO(L) \} \in  PCO(L)$, by means of Proposition 7, we have  $supp (\inf \{ PCO(L) \})   \in  CO(supp(L))$.
 It means $ \inf\{CO(supp(L))\}  \subseteq  supp (\inf \{ PCO(L) \})$.
According to the premise $ supp(L_{H_{n}}) = \inf \{ CO(supp(L)) \} $,
we have $ supp(L_{H_{n}}) \subseteq supp (\inf \{ PCO(L) \})$.
In addition, Lemma 2 implies $supp( L_{\widetilde{H}} ) = supp( L_{H_{n}} ) $.
Thus, we obtain $ supp( L_{\widetilde{H}} ) \subseteq supp (\inf \{ PCO(L) \})$.
\par
Therefore,  $ supp( L_{\widetilde{H}} ) = supp( \inf \{ PCO(L)\}) $.
\end{proof}

\begin{theorem}
  $L_{\widetilde{H}} = \inf \{ PCO(L) \}$.
\end{theorem}
\begin{proof}
We have shown that $L_{\widetilde{H}} \in PCO(L)$ in Lemma 3.
For contradiction, assume $L_{\widetilde{H}} \neq \inf \{ PCO(L)\} $.
Then there must exist $L^{'} \in PCO(L)$,  $s_{1} \in \Sigma^{*}$ and $\sigma_{1} \in \Sigma$, such that
 $\frac{L^{'}(s_{1}\sigma_{1})}{L^{'}(s_{1})} < \frac{L_{\widetilde{H}}(s_{1}\sigma_{1})}{L_{\widetilde{H}}(s_{1})}$.
 We deduce contradiction by dividing the following two cases.
 \begin{enumerate}
   \item
   $\sigma_{1} \in \Sigma_{uc}$
   \par
   $L_{\widetilde{H}} \in PCO(L)$ means that $L_{\widetilde{H}}$ is probabilistic controllable. Then we have
   $\frac{L_{\widetilde{H}}(s_{1}\sigma_{1})}{L_{\widetilde{H}}(s_{1})} = \frac{L_{G_{n}}(s_{1}\sigma_{1})}{L_{G_{n}}(s_{1})} $.
   Thus, $\frac{L^{'}(s_{1}\sigma_{1})}{L^{'}(s_{1})} < \frac{L_{G_{n}}(s_{1}\sigma_{1})}{L_{G_{n}}(s_{1})}$, which implies that
   $L^{'}$ is not probabilistic controllable.
   This  contradicts $L^{'} \in  PCO(L)$.

   \item
   $\sigma_{1} \in \Sigma_{c}$
   \par

   Suppose $\delta_{\widetilde{H}}(x_{0,\widetilde{H}}, s_{1}) = \delta_{G_{n}}(x_{0,G_{n}}, s_{1})= x_{1}$.
   $\frac{L^{'}(s_{1}\sigma_{1})}{L^{'}(s_{1})} < \frac{L_{\widetilde{H}}(s_{1}\sigma_{1})}{L_{\widetilde{H}}(s_{1})}$ implies that $L_{\widetilde{H}}(s_{1}\sigma_{1})  > 0$.
    According to Algorithm 3,
    we have $k_{1} = \max_{j = 1}^{q} \{ k_{j}^{1} \} > 0$.
  Without loss of generality,
  let $\frac{\rho_{\widetilde{H}}(x_{1},\sigma_{1})}{\rho_{G_{n}}(x_{1},\sigma_{1})} = k_{1} = k_{2}^{1} = \frac{\rho_{H_{n}}(x_{2},\sigma_{1})}{\rho_{G_{n}}(x_{2},\sigma_{1})}$.
  It means that $\rho_{\widetilde{H}}(x_{2},\sigma_{1}) = \rho_{H_{n}}(x_{2},\sigma_{1})$, and
  there exists $X_{1} \in X_{obs}$, such that $x_{1}, x_{2} \in X_{1}$.
  Suppose $\delta_{\widetilde{H}}(x_{0,\widetilde{H}}, s_{2}) = \delta_{G_{n}}(x_{0,G_{n}}, s_{2})= x_{2}$.
  Then we have $P(s_{1}) = P(s_{2})$, $s_{1}\sigma_{1},s_{2}\sigma_{1} \in supp(L_{\widetilde{H}})$, and
  $ \frac{L_{\widetilde{H}}(s_{2}\sigma_{1})}{L_{\widetilde{H}}(s_{2})} = \frac{L_{H_{n}}(s_{2}\sigma_{1})}{L_{H_{n}}(s_{2})}$.
  According to Lemma 4, we also have $s_{1}\sigma_{1},s_{2}\sigma_{1} \in supp(L^{'})$.
  By means of the probabilistic observability of $L_{\widetilde{H}}$ and $L^{'}$, we obtain
   $$\frac{L_{\widetilde{H}}(s_{1}\sigma_{1})}{L_{\widetilde{H}}(s_{1})} * \frac{L_{G_{n}}(s_{2}\sigma_{1})}{L_{G_{n}}(s_{2})} =
     \frac{L_{\widetilde{H}}(s_{2}\sigma_{1})}{L_{\widetilde{H}}(s_{2})} * \frac{L_{G_{n}}(s_{1}\sigma_{1})}{L_{G_{n}}(s_{1})} $$
   and
   $$ \frac{L^{'}(s_{1}\sigma_{1})}{L^{'}(s_{1})} * \frac{L_{G_{n}}(s_{2}\sigma_{1})}{L_{G_{n}}(s_{2})} =
     \frac{L^{'}(s_{2}\sigma_{1})}{L^{'}(s_{2})} * \frac{L_{G_{n}}(s_{1}\sigma_{1})}{L_{G_{n}}(s_{1})}. $$
   Since $\frac{L^{'}(s_{1}\sigma_{1})}{L^{'}(s_{1})} < \frac{L_{\widetilde{H}}(s_{1}\sigma_{1})}{L_{\widetilde{H}}(s_{1})}$
   and $ \frac{L_{\widetilde{H}}(s_{2}\sigma_{1})}{L_{\widetilde{H}}(s_{2})} = \frac{L_{H_{n}}(s_{2}\sigma_{1})}{L_{H_{n}}(s_{2})}$,
    as mentioned before,
   we obtain
   \begin{equation}
     \frac{L^{'}(s_{2}\sigma_{1})}{L^{'}(s_{2})}  < \frac{L_{H_{n}}(s_{2}\sigma_{1})}{L_{H_{n}}(s_{2})}.
   \end{equation}
   We show that Equation (51) is impossible by dividing the following two cases:
   \begin{enumerate}
     \item
   If $s_{2}\sigma_{1} \in supp(L_{H_{n}}) \backslash supp(L)$, according to the construction algorithm for $H_{n}$ (see Appendix A),
   $\frac{L_{H_{n}}(s_{2}\sigma_{1})}{L_{H_{n}}(s_{2})} = 0^{+}$, where $0^{+}$ is viewed as the ``minimum" positive number.
   Hence, Equation (51) is impossible.

     \item
     If $s_{2}\sigma_{1} \in supp(L)$, then $L_{H_{n}}(s_{2}\sigma_{1}) = L(s_{2}\sigma_{1})$ and $L_{H_{n}}(s_{2}) = L(s_{2})$.
        Equation (51) is equivalent to $\frac{L^{'}(s_{2}\sigma_{1})}{L^{'}(s_{2})}  < \frac{L(s_{2}\sigma_{1})}{L(s_{2})}$.
        On the other hand, $L^{'} \in  PCO(L) $ implies
        $\frac{L(s_{2}\sigma_{1})}{L(s_{2})} \leq \frac{L^{'}(s_{2}\sigma_{1})}{L^{'}(s_{2})}$.
          Hence, Equation (51) is impossible.
   \end{enumerate}
 \end{enumerate}
\end{proof}

\section{Conclusions}
In this paper, we have formulated a comprehensive theory for the supervisory control problem of PDESs with the assumptions that the supervisor is probabilistic
 and has a partial observation. The main contributions of this paper are as follows.
 \begin{enumerate}
   \item
     The partial observation probabilistic supervisor has been defined as
     a set of probability distributions on the control patterns, called as the probabilistic P-supervisor.
     The equivalence between the  probabilistic P-supervisor and the scaling-factor function has been demonstrated.
     As a result, the scaling-factor function could be viewed as a compact form of the probabilistic P-supervisor.
     \item
     The notions of the probabilistic controllability and observability,
     and their polynomial verification algorithms have been proposed.
     The probabilistic controllability and observability theorem has been put forward,
     in which the probabilistic controllability and observability are demonstrated to be the necessary and sufficient conditions for the existence of the probabilistic P-supervisors.
     Moreover, the probabilistic P-supervisors synthesizing approach also has been presented.

     \item
     The optimal control problem of PDESs has been considered.
     The infimal probabilistic controllable and observable superlanguage,
     as the solution of optimal control problem of PDESs,
     has been introduced and computed.
 \end{enumerate}

  \par
  The centralized control of PDESs has been considered in this paper.
  A further issue to be considered is the decentralized control of PDESs.
  Moreover, Lin \cite{networked} investigated the control problem of networked DES that
  deals with the communication losses and delays.
  However, \cite{networked} does not consider the probabilities of the communication losses and delays,
  which might exist and could be obtained in probabilistic systems.
  The control of networked PDESs that deals with the probabilistic communication losses and delays could be another challenge.
  These two aforementioned problems should be worthy of consideration in subsequent work.

\begin{appendices}
\section{ Construction algorithm for $G_{n}$ and $H_{n}$}

\begin{algorithm}
  For the sets $PCO(L)$ and $CO(supp(L))$ defined before,
  given automaton $H_{s} =  \{ Q_{2}, q_{0,2}, \Sigma, \delta_{2}  \}$ such that $L_{H_{s}} = \inf \{ CO(supp(L)) \}$,
  and probabilistic automata $G = \{ Q_{1}, q_{0,1}, \Sigma, \delta_{1}, \rho_{1}  \} $ and $H = \{ Q_{3}, q_{0,3}, \Sigma, \delta_{3},  \rho_{3} \}$, such that
  $L_{H} = L$ and $L_{G} = M$.
  \begin{enumerate}
    \item
    Construct probabilistic automata $H^{'} = \{ X_{H^{'}}, x_{0,H^{'}},\Sigma, $ $ \delta_{H^{'}}, \rho_{H^{'}}  \} $,
    $H_{s}^{'} = \{ X_{H^{'}_{s}}, x_{0,H^{'}_{s}},\Sigma, \delta_{H^{'}_{s}}, \rho_{H^{'}_{s}} \}$ and
    $G^{'} = \{ X_{G^{'}}, x_{0,G^{'}},\Sigma,  \delta_{G^{'}}, \rho_{G^{'}} \} $,
    such that $L_{H^{'}} = L_{H}$, $supp(L_{H_{s}^{'}}) = L_{H_{s}}$, $ L_{G^{'}} = L_{G}$, and
    $H^{'} \sqsubseteq H_{s}^{'} \sqsubseteq G^{'}$. Specifically, do the following.
        \begin{enumerate}
         \item
         Let $logic(H^{'}) = logic(G) \times H_{s} \times logic(H)$ first, then
          specify a transition probability for each transition in $logic(H^{'})$.
          Specifically,
         let $\rho_{H^{'}}((q_{1},q_{2},q_{3}), \sigma) = \rho_{3}(q_{3},\sigma)$.
         Thus, obtain the probabilistic automaton $H^{'}$.

         \item
         Examine each state $q_{3}$ of $logic(H)$ and add a self-loop for each event that is not defined at $q_{3}$, and call the result $logic(H)^{sl}$.

         \item
          Let $0^{+}$ denotes the positive number that is less than any a given positive number.  That is, $0^{+}$ could be viewed as the ``minimum" positive number.

         \item
         Let $logic(H^{'}_{s}) = logic(G) \times H_{s} \times logic(H)^{sl}$ first,
         then specify  a probability for each transition in $logic(H^{'}_{s})$.
         Specifically,
         let $\rho_{H^{'}_{s}}((q_{1},q_{2},q_{3}), \sigma) = \\ \rho_{3}(q_{3},\sigma)$, if $\delta_{3}(q_{3},\sigma)!$; otherwise,
         $\rho_{H^{'}_{s}}((q_{1},q_{2},q_{3}), \\ \sigma) = 0^{+}$.
         Thus, obtain the probabilistic automaton $H^{'}_{s}$.

         \item
         Examine each state $q_{2}$ of $H_{s}$ and add a self-loop for each event that is not defined at $q_{2}$, and call the result $H^{sl}_{s}$.

          \item
         Let $logic(G^{'}) = logic(G) \times H_{s}^{sl} \times logic(H)^{sl}$ first, then specify  a probability for each transition in $logic(G^{'})$.
         Specifically,
         let $\rho_{G^{'}}((q_{1},q_{2},q_{3}), \sigma) = \rho_{1}(q_{1},\sigma)$.
         Thus, obtain the probabilistic $G^{'}$.

    \end{enumerate}

    \item
    Construct the observers for $logic(G^{'})$ of $logic(H^{'}_{s})$, namely, $Obs(logic(G^{'}))$ and $Obs(logic(H^{'}_{s}))$,
    such that $Obs(logic(H^{'}_{s})) \sqsubseteq  Obs(logic(G^{'}))$.
    Specifically, do the following.
        \begin{enumerate}
              \item
                Obtain the observers for $logic(G^{'})$ of $logic(H^{'}_{s})$, denoted by $Obs_{G}$ and $Obs_{H}$, respectively.
              \item
                Examine each state $q_{obs}$ of $Obs_{H}$ and add a self-loop for each event that is not defined at $q_{obs}$, and call the result $Obs_{H}^{sl}$.
              \item
               Let $Obs(logic(H^{'}_{s})) = Obs_{G} \times Obs_{H}$, and $ Obs(logic(G^{'})) = Obs_{G} \times Obs_{H}^{sl}$.
          \end{enumerate}

    \item
    Suppose $Obs(logic(G^{'})) = \{X_{obs},  x_{0,obs}, \Sigma_{o}, \delta_{obs} \}$.
    Then construct the normal probabilistic automaton $G_{n}$, such that $L_{G_{n}} = L_{G^{'}}$.
    Specifically,
      let $G_{n} = \{ X_{G_{n}}, x_{0,G_{n}},\Sigma,  \delta_{G_{n}}, \rho_{G_{n}} \}$,
      where $X_{G_{n}} = X_{G^{'}} \times X_{obs}$ and $x_{0,G_{n}} = (x_{0,G^{'}}, x_{0,obs})$.
      The transition function $\delta_{G_{n}}$ is defined as follows.
      \begin{align}
          \delta_{G_{n}}((x^{'},x_{obs}),\sigma) =
              \begin{cases}
                      (\delta_{G^{'}}(x^{'},\sigma), &  \delta_{obs}(x_{obs},\sigma)),
                        \\ & \text{ if }    \delta_{G^{'}}(x^{'},\sigma)! \text{ and } \sigma \in \Sigma_{o},     \\
                      (\delta_{G^{'}}(x^{'},\sigma), &  x_{obs}),
                          \\ & \text{ if }    \delta_{G^{'}}(x^{'},\sigma)! \text{ and } \sigma \in \Sigma_{uo},
              \end{cases}
      \end{align}
      The transition probability function $\rho_{G_{n}}$ is defined as follows.
      \begin{equation}
            \rho_{G_{n}}((x^{'},x_{obs}), \sigma) = \rho_{G^{'}}(x^{'}, \sigma).
       \end{equation}

     \item
    Similarly, construct the normal probabilistic automaton $H_{n}$, such that $L_{H_{n}} = L_{H^{'}_{s}}$.

  \end{enumerate}

\end{algorithm}

\par
\par
    The main idea of the step 1) of Algorithm 4 that constructs three probabilistic automata subject to subautomaton relation is from \cite{desbook} (page 87),
    in which a general procedure
     to build two non-probabilistic automata subject to subautomaton relation was presented.
    The method of the step 2) of Algorithm 4 is directly from \cite{desbook} (page 87).
    The main idea of the step 3) of Algorithm 4 that makes a normalization for $G^{'}$ is from \cite{Cho} and \cite{Takai}.
    Since $H^{'}_{s} \sqsubseteq G^{'}$ and $Obs(logic(H^{'}_{s})) \sqsubseteq  Obs(logic(G^{'}))$,
    the step 3) and 4) of Algorithm 4 preserve the subautomaton relation.
    Hence, $H_{n} \sqsubseteq G_{n}$.
    \par
    In the following, we would like to prove the correctness of Algorithm 4. Since the step 2) is similar to the step 1), and the step 4) is similar to step 3),
    we only present the proofs for the step 1) and step 3).
    \par
    \begin{proposition}
      After executing the step 1) of Algorithm 4,
      $L_{H^{'}} = L_{H}$, $supp(L_{H_{s}^{'}}) = \widehat{L}_{H_{s}}$, $L_{H_{s}^{'}}(s) = L_{H^{'}}(s)$ for $\forall s \in supp(L_{H^{'}})$,
      $ L_{G^{'}} = L_{G}$, and     $H^{'} \sqsubseteq H_{s}^{'} \sqsubseteq G^{'}$.
    \end{proposition}
    \begin{proof}
    We first show that after executing the step 1) of Algorithm 4, $L_{H^{'}} = L_{H}$.
    By $L = L_{H}$, $M = L_{G}$ and $\widehat{L}_{H_{s}} = \inf\{CO(supp(L))\}$, we have $\widehat{L}_{logic(H)} \subseteq \widehat{L}_{H_{s}} \subseteq \widehat{L}_{logic(G)}$.
    Since $logic(H^{'}) = logic(G) \times H_{s} \times logic(H)$, we obtain $\widehat{L}_{logic(H^{'})} = \widehat{L}_{logic(H)}$.
    That is, $ supp(L_{H^{'}}) = supp(L_{H})$.
    Furthermore, $logic(H^{'}) = logic(G) \times H_{s} \times logic(H)$ also implies that if there exists an $s \in \Sigma^{*}$ such that $\delta_{3}(q_{0,3},s) = q_{3}$, then
    there must exist a $(q_{1},q_{2},q_{3}) \in X_{H^{'}}$, such that $\delta_{H^{'}}((q_{0,1}, q_{0,2}, q_{0,3}),s) = (q_{1},q_{2},q_{3})$.
    In addition, we have $\rho_{H^{'}}((q_{1},q_{2},q_{3}), \sigma) = \rho_{3}(q_{3},\sigma)$.
    By induction on the length of the events sequence, it is easy to prove that $L_{H^{'}} = L_{H}$.
    \par
    We then show that after executing the step 1) of Algorithm 4, $supp(L_{H_{s}^{'}}) = \widehat{L}_{H_{s}}$ and $L_{H_{s}^{'}}(s) = L_{H^{'}}(s)$ for $\forall s \in supp(L_{H^{'}})$.
    According to the construction method of $login(H)^{sl}$, we have $\widehat{L}_{login(H)^{sl}} = \Sigma^{*}$.
    Hence, $ \widehat{L}_{H_{s}} \subseteq \widehat{L}_{logic(G)} \subseteq \widehat{L}_{login(H)^{sl}}$.
    Since $logic(H^{'}_{s}) = logic(G) \times H_{s} \times logic(H)^{sl}$, thus $\widehat{L}_{logic(H^{'}_{s})} = \widehat{L}_{H_{s}}$.
    That is, $supp(L_{H^{'}_{s}}) = \widehat{L}_{H_{s}}$.
    Similar to the proof of $L_{H^{'}} = L_{H}$ presented above,  we can prove that
     $L_{H_{s}^{'}}(s) = L_{H}(s)$ for $\forall s \in supp(L_{H})$.
    That is, $L_{H_{s}^{'}}(s) = L_{H^{'}}(s)$ for $\forall s \in supp(L_{H^{'}})$.
    \par
    Similarly, we could prove that after executing the step 1) of Algorithm 4, $ L_{G^{'}} = L_{G}$.
    \par
    The remain is to prove that after executing the step 1) of Algorithm 4, $H^{'} \sqsubseteq H_{s}^{'} \sqsubseteq G^{'}$.
    We only prove $H^{'} \sqsubseteq H_{s}^{'}$, as $H_{s}^{'} \sqsubseteq G^{'}$ can be proved similarly.
    Since
    $logic(H^{'}) = logic(G) \times H_{s} \times logic(H),$ and
    $logic(H^{'}_{s}) = logic(G) \times H_{s} \times logic(H)^{sl},$
    we have $$x_{0,H^{'}} = x_{0,H^{'}_{s}} = (q_{0,1},q_{0,2},q_{0,3}),$$ and $X_{H^{'}} \subseteq X_{H^{'}_{s}}$, and
    for $\forall (q_{1},q_{2},q_{3}) \in X_{H^{'}}$ and $\sigma \in \Sigma$, if $\delta_{H^{'}}((q_{1},q_{2},q_{3}),\sigma)$!,
    $$\delta_{H^{'}_{s}}((q_{1},q_{2},q_{3}),\sigma) = \delta_{H^{'}}((q_{1},q_{2},q_{3}),\sigma). $$
    Since
    $\rho_{H^{'}_{s}}((q_{1},q_{2},q_{3}), \sigma) =  \rho_{3}(q_{3},\sigma)$, if $\delta_{3}(q_{3},\sigma)!$,
         and
     $ \rho_{H^{'}}((q_{1},q_{2},q_{3}), \sigma) = \rho_{3}(q_{3},\sigma)$,
     thus we could obtain
     $$ \rho_{H^{'}}((q_{1},q_{2},q_{3}), \sigma) \leq  \rho_{H^{'}_{s}}((q_{1},q_{2},q_{3}), \sigma),$$
     for $\forall (q_{1},q_{2},q_{3}) \in X_{H^{'}}$ and $\sigma \in \Sigma$.
     Therefore, by the definition of subautomaton, $H^{'} \sqsubseteq H_{s}^{'}$.
    \end{proof}
    \par
    \begin{proposition}
      After executing the step 3) of Algorithm 4, $logic(G_{n})$ is a normal automaton, and $L_{G_{n}} = L_{G^{'}}$.
    \end{proposition}
    \begin{proof}
        The proof of the normality of $logic(G_{n})$ can refer to \cite{Cho} and \cite{Takai}.
        We only prove $L_{G_{n}} = L_{G^{'}}$.
        According to the definitions of observer automaton and $G_{n}$, $\delta_{G^{'}}(x^{'},\sigma)!$ always implies $\delta_{obs}(x_{obs},\sigma)!$,
         for $(x^{'},x_{obs}) \in X_{G_{n}}$.
        Thus, $supp(L_{G^{'}}) = supp(L_{G_{n}})$.
        By the definition of $G_{n}$,
        if there exists an $s \in \Sigma^{*}$ such that $\delta_{G^{'}}(x_{0,G^{'}},s) = x^{'}$, then
    there must exist an $(x^{'},x_{obs}) \in X_{G_{n}}$, such that $\delta_{G_{n}}((x_{0,G^{'}}, x_{0,obs}),s) = (x^{'},x_{obs})$.
         Moreover, we have $\rho_{G_{n}}((x^{'},x_{obs}), \sigma) = \rho_{G^{'}}(x^{'}, \sigma)$.
    By induction on the length of the events sequence, it is easy to prove that $L_{G^{'}} = L_{G_{n}}$.
    \end{proof}

\end{appendices}

\section*{Acknowledgements}

 This work is supported in part by the National
Natural Science Foundation of China (Nos. 61572532, 61272058),  the Natural Science
Foundation of Guangdong Province of China (No. 2017B030311011), and the Fundamental Research Funds for the Central Universities of China (No. 17lgjc24),
and Deng is supported partially by the Natural Science Foundation of Guangdong Province of China (Nos. 2017A030310583).

%

  \newpage

%

\end{document}